\documentclass[10pt]{article}

\usepackage{scai}

\usepackage[authoryear, sort&compress, round]{natbib}
\usepackage{fontawesome5} % Provides the \faGithub icon
\usepackage{hyperref}     % Provides hyperlinking (\href and \url)
\usepackage{xcolor}       % Optional: for custom colors
\usepackage{booktabs}
\usepackage{multirow}
\usepackage{subcaption}
\usepackage{algorithm}
\usepackage{algpseudocode}
\usepackage{cleveref}
\usepackage{bm}
\usepackage{lipsum}
\usepackage{tikzducks}
\usepackage{adjustbox}
\usepackage{makecell}
\usepackage{multirow}
\usepackage{tabularx}
\usepackage{longtable}
\usepackage{xcolor}
\usepackage{placeins}
\usepackage{wrapfig}
\usepackage{amsthm}

\hypersetup{
    colorlinks=true,
    linkcolor=blue,
    urlcolor=blue,
}

\annotator{alice}{red}
\annotator{bob}{blue}

\newcommand{\taumodel}{\tau_{\mathrm{WK}}}
\newcommand{\tauhat}{\hat{\tau}_{\mathrm{WK}}}
\newcommand{\tauwave}{\tau_{\mathrm{wave}}}
\newcommand{\Qin}{Q_{\mathrm{in}}}

\newcommand{\TauTestWindows}{34869}
\newcommand{\TauTestPatients}{475}
\newcommand{\TauValidationWindows}{37888}
\newcommand{\TauValidationPatients}{474}
\newcommand{\TauBaselineTestMAE}{0.434}
\newcommand{\TauTestMAE}{0.297}
\newcommand{\TauTestPearson}{0.678}
\newcommand{\COTestPatients}{120}
\newcommand{\COTestWindows}{15509}
\newcommand{\COTestMAE}{1.059}
\newcommand{\COTestPE}{47.1\%}
\usepackage[english]{babel}
\newtheorem{proposition}{Proposition}

\vspace{-5pt}
\title{HIPNO: Symmetry-Aware Physics-Informed Neural Operators for Noninvasive Hemodynamic Inference}
\contribnote{equal}{Equal contribution.}              % auto-assigned: *
\contribnote{corr}{\ddag}{Corresponding author.}      % explicit symbol: ‡

\scaiauthor{1,equal}{Yunbei Pan}
\scaiauthor{1,2,equal}{Jiahang Sha}
\scaiauthor{1}{Simon A. Lee}
\scaiauthor{3}{Maxime Cannesson}
\scaiauthor{1,2}{Wei Wang}
\scaiauthor{1,4,corr}{Jeffrey N. Chiang}

\affiliation{1}{Department of Computational Medicine, UCLA}
\affiliation{2}{Scalable Analytics Institute (ScAi), Department of Computer Science, UCLA}
\affiliation{3}{Department of Anesthesiology \& Perioperative Medicine, UCLA}
\affiliation{4}{Department of Neurosurgery, UCLA}

\scaikeywords{physics-informed machine learning, neural operators, symmetry-aware learning, structural identifiability, noninvasive hemodynamic monitoring, physiological inverse problems, Windkessel model, cardiovascular physiology}

\metadata[Date]{July 22, 2026}
\metadata[Github]{\href{https://github.com/ybpan16/HIPNO}{\faGithub~https://github.com/ybpan16/HIPNO}}
\begin{abstract}
Continuous hemodynamic monitoring guides treatment decisions in surgery and intensive care. However, gold-standard signals are only measured in severe cases due to risks associated with invasive measurement. In this work, we introduce HIPNO (Hemodynamic Inference via Physics-informed Neural Operators) to recover hemodynamic state from ubiquitous, non-invasive signals and expand access to advanced monitoring. HIPNO addresses a problem of scale symmetry in physics-informed hemodynamic inference, where different combinations of flow, resistance, and compliance can generate the same observed pressure. We identify the symmetry group of the observation model and parameterize the network in its quotient space.
For the 3-element Windkessel model, the quotient coordinates are the compliance-normalized flow $U=Q/C$, the decay time constant $\taumodel=R_2C$, and the characteristic-impedance coordinate $\kappa=R_1C$. Across 945499 intraoperative windows from 2562 patients, HIPNO predicts $\tauwave$, a proxy for vascular decay derived from pressure, with 32\% lower error on the log scale than a population baseline while preserving mean arterial pressure accuracy. Because vascular decay and flow drive occupy separate coordinates, counterfactual perturbations produce the expected directional responses in at least 90\% of windows in almost all prespecified scenarios, a separation unavailable to pressure-only baselines. The coordinates are also used as inputs to a calibration model for monitored cardiac output. Finally, the formulation identifies the external compliance or flow reference required to recover absolute physical scale.
\end{abstract}

\begin{document}
\maketitle

\section{Introduction}
\label{sec:intro}

Continuous vital sign monitoring guides critical real-time decisions in intensive care and surgery, from fluid resuscitation to vasopressor titration. While electrocardiography (ECG) and photoplethysmography (PPG) are monitored routinely and non-invasively, the parameters that most directly characterize a patient’s hemodynamic state remain difficult to acquire continuously. Arterial blood pressure (ABP) is invasively measured via an arterial line catheter but can be approximated intermittently using an inflatable cuff (NIBP). Central venous pressure (CVP), cardiac output (CO), and vascular tone require central catheters positioned in major blood vessels near the heart. Because these invasive procedures carry material risks such as thrombosis and hemorrhage, clinical guidelines restrict their use primarily to patients who are already critically unstable \citep{cannesson2011hemodynamic}. Inferring these key targets from routinely available signals (ECG and PPG) would extend hemodynamic surveillance to broader patient populations and allow timely intervention before overt deterioration occurs.

While data-driven models have shown success in reconstructing blood pressure from noninvasive signals \citep{ibtehaz2022ppg2abp,kundu2025inn,chen2026versatile}, continuous volumetric flow and vascular tone remain largely unaddressed, primarily due to the scarcity of clinical ground truth measurements. For instance, cardiac output (CO) is seldom measured outside cardiac surgery, leaving ABP as an imperfect surrogate for overall hemodynamic state \citep{cannesson2011hemodynamic}. Fortunately, the physical governing equations coupling pressure and flow are well established \citep{westerhof2009arterial}. Physics-informed models leverage these underlying domain dynamics to compensate for data scarcity by imposing physical constraints that observational signals alone cannot provide \citep{RAISSI2019686,li2024bp,sel2023physics}.

This physics-guided inference comes with a structural non-identifiability obstacle. Even when a model fits the observed data and satisfies the governing equations, the underlying physiological state can remain mathematically undetermined. Specifically, when the observation model is invariant under a continuous transformation group, a family of distinct latent states yields identical observable signals. A prominent example is scale symmetry, where proportional rescaling of cardiac flow, vascular resistance, and arterial compliance maps to the exact same pressure waveform. We eliminate such ambiguity by identifying the underlying symmetry group of the observation model and formulating our learning process directly on the invariant quotient space it defines.

We introduce HIPNO (Hemodynamic Inference via Physics-informed Neural Operators), an operator model designed to infer invasive hemodynamic quantities from noninvasive ECG and PPG signals. By isolating scale ambiguity in quotient coordinates, HIPNO disentangles vascular tone ($\taumodel$) from flow drive ($U$, the compliance-normalized flow that scales cardiac output) while preserving pressure reconstruction fidelity. Our main contributions are summarized as follows:
\begin{itemize}
    \item We formalize a pressure invariant scale symmetry governing the 3-element Windkessel dynamics (Proposition~\ref{prop:scale}) and derive its explicit quotient coordinates $U$, $\taumodel$, and $\kappa$.
    \item We construct a neural operator in this quotient space, coupling ECG/PPG encoders with a physics-conditioned continuous trunk with staged supervision.
    \item We validate the coordinate representation by predictions on held-out patients and by counterfactual perturbations. Meanwhile, we characterize the minimal external measurement required to determine the absolute physical scale.
\end{itemize}
\begin{figure*}[t]
    \centering
    \includegraphics[width=\textwidth]{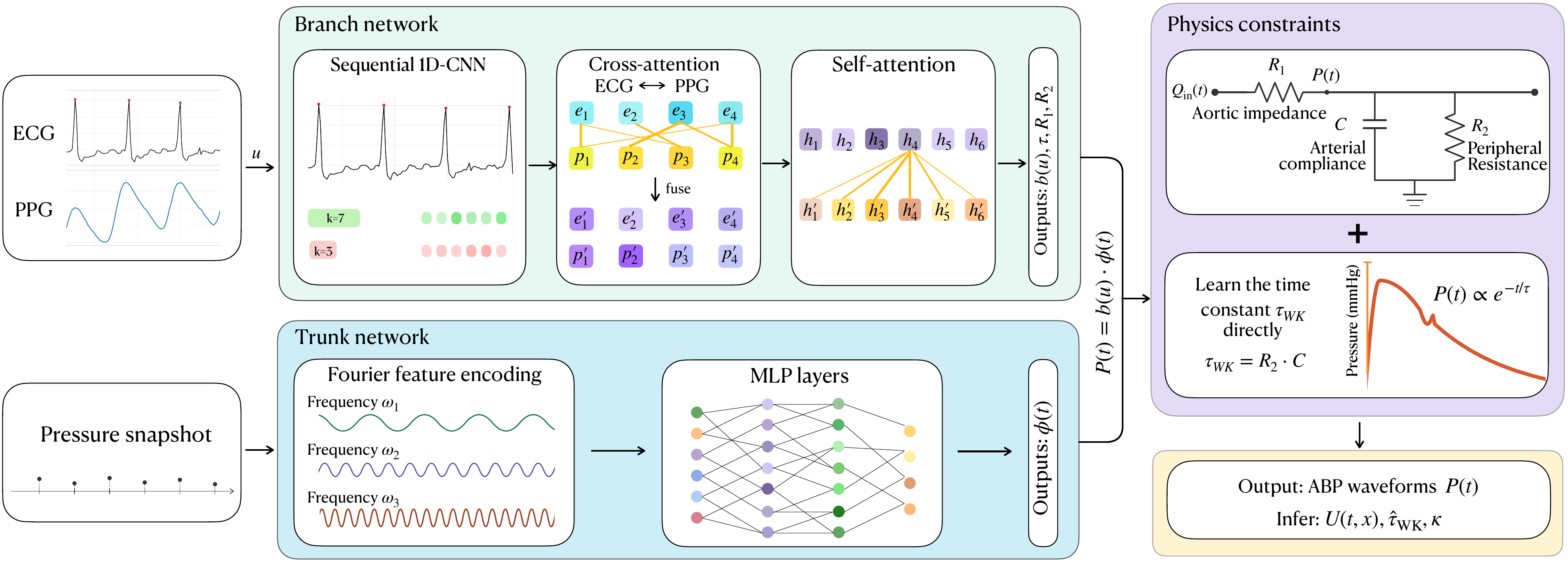}
    \caption{HIPNO architecture and information pathway. ECG, PPG, and cuff variables are clinically available inputs. Invasive ABP and the derived $\tauwave$ proxy are training targets used during development only. Monitor CO and EHR covariates are used only after the operator is frozen.}
    % \Description{The HIPNO flowchart shows ECG and PPG inputs entering a convolutional, cross-attention, and self-attention encoder network. Pressure query coordinates are passed to a Fourier feature trunk network. The coupled outputs feed Windkessel constraints and produce an arterial pressure waveform and latent hemodynamic quantities. The figure distinguishes deployment inputs from training targets used during development and from downstream calibration variables.}
    \label{fig:pipeline}
\end{figure*}

\section{Related Work}

\textbf{Symmetry-aware physics-informed learning.}
Exploiting known symmetries as an inductive bias is well established in machine learning, spanning group equivariant networks \citep{cohen2016group,weiler2019general} and geometric deep learning \citep{bronstein2021geometric} to physical science models \citep{villar2021scalars}. 
% Separately, work on identifiability shows that parameters in inverse problems are recoverable only up to invariant symmetry transformations \citep{bellman1970structural,khemakhem2020variational}. While physics-informed networks and operators constrain predictions with differential equations \citep{RAISSI2019686,kovachki2023neural,lu2021learning}, these constraints often leave underlying symmetries intact. Cardiovascular models combining 1D flow or Windkessel dynamics with neural prediction learn raw parameters jointly, leaving the scale ambiguity formalized in Section~\ref{sec:identifiability} unaddressed \citep{li2024bp,sel2023physics,osman2026cardiovascular}. Rather than embedding equivariance or merely characterizing residual ambiguity, HIPNO operates directly in the symmetry group's quotient space.
% 
A complementary line of work on identifiability in latent variable and inverse problems shows that parameters are generally recoverable only up to the transformations a symmetry leaves invariant \citep{bellman1970structural,khemakhem2020variational}. Physics-informed neural networks and operator models constrain predictions with differential equations \citep{RAISSI2019686,kovachki2023neural,lu2021learning}. However, those constraints can leave such symmetries intact in the inverse problem. Cardiovascular models that combine one-dimensional flow or Windkessel dynamics with neural prediction learn resistance, compliance, and flow together in their original, unnormalized coordinates, and so do not address the scale ambiguity that we formalize in Section~\ref{sec:identifiability} \citep{li2024bp,sel2023physics,osman2026cardiovascular}. Rather than embedding equivariance into feature maps or only characterizing the residual ambiguity, HIPNO works directly in the quotient space of that symmetry group.

\textbf{Noninvasive pressure reconstruction.}
Pulse timing methods map delays between electrical activity and peripheral pulse arrival to scalar pressure estimates \citep{ding2019pulse,mukkamala2015toward, zhou2026physiology}. Learned waveform models reconstruct continuous ABP from PPG or multimodal signals \citep{ibtehaz2022ppg2abp,kundu2025inn,chen2026versatile}. While providing strong predictive baselines, their latent variables are unconstrained representations of pressure.

\textbf{Pressure-derived time constants and CO.}
The Windkessel diastolic time constant can be estimated from invasive pressure under assumptions of negligible diastolic inflow and a specified venous pressure \citep{westerhof2009arterial,bikia2024novel}. Pulse contour monitors estimate CO from arterial pressure and patient covariates \citep{slagt2014systematic,manecke2005edwards}. Correlation measures association at the cohort level, whereas agreement in method comparison additionally depends on bias and dispersion within each patient \citep{critchley1999meta}.

\section{Methods}
\label{sec:methods}

\subsection{Observation Model and Outputs}
\label{sec:formulation}

For a window $[0,T]$, let $S_{\mathrm{ECG}}$ and $S_{\mathrm{PPG}}$ denote ECG and PPG signals. Their joint input is
\begin{equation}
    \mathbf{u}(t)=\bigl(S_{\mathrm{ECG}}(t),S_{\mathrm{PPG}}(t)\bigr)
    \in L^2([0,T],\mathbb{R}^2).
\end{equation}
Each window is also paired with a cuff-derived vector $\mathbf{b}\in\mathbb{R}^6$. Its entries are normalized SBP, DBP, MAP, pulse pressure, the signed time offset from the window midpoint to the nearest NIBP observation, and a validity indicator. Invalid cuff values are replaced by zero and marked by the indicator. Each window is matched to the nearest NIBP observation recorded within 300s. Appendix Section~\ref{app:architecture} gives the normalization constants.

HIPNO maps $(\mathbf{u},\mathbf{b})$ to reconstructed pressure $P(t)$, latent arterial fields, and positive Windkessel coordinates $(\kappa,\taumodel)$. The fields are cross-sectional area $A(t,x)$, flow normalized by compliance $U(t,x)$, and transmural pressure $P_{\mathrm{tm}}(t,x)$ on the normalized arterial segment $x\in[0,L]$ with $L=1$.

The core operator can be written as
\begin{equation}
    \mathcal G_{\Theta}(\mathbf u,\mathbf b)(t,x)
    =\{P(t),A(t,x),U(t,x),P_{\mathrm{tm}}(t,x),\kappa,\taumodel\}.
    \label{eq:operator}
\end{equation}
The coordinate arguments distinguish the formulation from a model with a fixed output resolution. HIPNO is evaluated at arbitrary times through a learned cardiac phase and at arbitrary locations through a normalized arterial coordinate. Both training and evaluation use 250 temporal points per 10s window and 16 output locations in space.

\subsection{Scale Ambiguity and Invariant Coordinates}
\label{sec:identifiability}

At the distal arterial boundary, the 3-element Windkessel equations are
\begin{align}
    \dot P_C(t) &= \Qin(t)/C-\{P_C(t)-P_v\}/(R_2C), \\
    P_{\mathrm{WK}}(t) &= P_C(t)+R_1\Qin(t),
\end{align}
where $P_C$ is the capacitive pressure and $P_v$ is venous back pressure \citep{belz1995elastic,westerhof2009arterial,segers1999pulmonary}. The model predicts $P_v$ for each window from the global latent representation. Its soft prior and bounds are given in Appendix~\ref{app:architecture}.

Two parameter settings are observationally equivalent when they generate the same pressure for the same initial condition over the observation interval. Any continuous family of such settings creates a structural ambiguity. The following result identifies one such family, independently of the choice of architecture or optimizer.

\begin{proposition}
\label{prop:scale}
For any $s>0$, define $C'=sC$, $\Qin'=s\Qin$, $R_1'=R_1/s$, and $R_2'=R_2/s$. The Windkessel pressure is unchanged, so $(\Qin,R_1,R_2,C)$ cannot be uniquely recovered from pressure.
\end{proposition}

\begin{proof}
Under this substitution, $\Qin'/C'=\Qin/C$, $R_2'C'=R_2C$, and $R_1'\Qin'=R_1\Qin$. Every coefficient of the two boundary equations is therefore unchanged, so both settings generate the same $P_C$ from a common initial condition and thus the same $P_{\mathrm{WK}}$.
\end{proof}

This result motivates
\begin{equation}
    U(t,x)=\frac{Q(t,x)}{C},\quad
    \taumodel=R_2C,\quad
    \kappa=R_1C.
\end{equation}
At $x=L$, $U_L(t)=U(t,L)$ and the boundary equations become
\begin{align}
    \dot P_C(t) &= U_L(t)-\frac{P_C(t)-P_v}{\taumodel},
    \label{eq:windkessel_identifiable}\\
    P_{\mathrm{WK}}(t) &= P_C(t)+\kappa U_L(t).
    \label{eq:pabp_identifiable}
\end{align}
These coordinates quotient out the scale symmetry.

A pressure and residual loss is constant along this scale family. Priors or regularization can select a representative, but cannot identify the underlying physical values from pressure alone. Thus, working directly in the invariant coordinates makes the latent parameters interpretable, since they are the only quantities that pressure data can determine.

The transformed boundary model removes this scale degree of freedom because $U_L$, $\kappa$, and $\taumodel$ do not change with $s$. During systole, $U_L(t)$ is a learned field and $\kappa U_L(t)$ enters the pressure equation as a product, so that $\kappa U_L=R_1\Qin$ recovers the characteristic impedance term. The model predicts $\kappa$, $\taumodel$, $P_v$, and $P_C(0)$ per window, constrained by priors, the boundary ODE, and the arterial residuals.

Cardiac output follows
\begin{equation}
    \mathrm{CO}\approx\frac{(\mathrm{MAP}-P_v)C}{\taumodel},
    \label{eq:co}
\end{equation}
where $C$ remains unobserved under the pressure-only model.

During diastole, the approximation $U_L\approx0$ gives
\begin{equation}
    P_C(t)=\{P_0-P_v\}\exp(-t/\taumodel)+P_v.
    \label{eq:diastolic_decay}
\end{equation}
Thus $\taumodel$ is a diastolic decay time constant. Individual values of $R_2$ and $C$ remain unidentified. The proxy $\tauwave$ derived from pressure (the diastolic decay fit defined in Section~\ref{sec:data}) assumes negligible residual inflow and reliable valve closure localization. Sensitivity to the finite window range appears in Appendix Table~\ref{tab:tau_sensitivity}. The operator is then built in these coordinates (Sections~\ref{sec:architecture} and~\ref{sec:training}), and Section~\ref{sec:discussion} distills the full construction into a general four-step procedure.

\subsection{Neural Operator with a Learned Cardiac Phase}
\label{sec:architecture}

HIPNO uses separate convolutional Transformer encoders for ECG and PPG. A cross-attention decoder fuses their local tokens with a window representation. The model learns a monotone cardiac phase
\begin{equation}
    \dot\phi(t)=\omega_0\bigl[\operatorname{softplus}\{f_{\mathrm{cnn}}(S_{\mathrm{ECG}}(t))\}+\epsilon\bigr],
    \quad \phi(t)=\int_0^t\dot\phi(s)\,ds,
\end{equation}
where $\omega_0$ is initialized at $2\pi\times1.25$ Hz and $\epsilon=10^{-4}$. The cyclic coordinate $\theta(t)=\phi(t)\bmod 2\pi$ aligns within-beat morphology. The normalized coordinate $p(t)=\phi(t)/\phi(T)$ preserves progress across the 10s window.

\textbf{Modality encoding and fusion.}
Separate encoders preserve the morphology specific to each modality, such as QRS timing in ECG and pulse arrival shape in PPG, while cross-attention captures their complementary timing information. Local queries aligned to phase attend separately to the ECG and PPG encoder outputs before fusion into a representation varying over time. A separate learned query produces the global window representation $\mathbf z_w$. Appendix Section~\ref{app:architecture} gives the exact architecture hyperparameters.

\textbf{Generating fields from coordinates.}
The trunk network \citep{lu2021learning} maps Fourier features of phase and space to continuous arterial fields. It excludes normalized progress and local decoder context, so the physics branch is conditioned on phase, space, and the global window state, with no separate waveform decoder. The trunk is conditioned on $\mathbf z_w$ through feature-wise linear modulation. Appendix Section~\ref{app:architecture} gives the exact dimensions. The two trunk outputs generate
\begin{equation}
    A=\operatorname{softplus}(s_A h_A+b_A)+10^{-3},\quad
    U=(s_U h_U+b_U)S_U,
\end{equation}
where the per-window scales are positive and, with the biases, depend on $\mathbf z_w$. The full $U(t,x)$ field is shifted to a smooth latent conditioned inlet pulse, and its scale $S_U$ is penalized on the log scale toward a fixed value, which sets the otherwise arbitrary normalization of the residual flow. The Windkessel head predicts $\kappa$, $\taumodel$, and $P_v$ for each window, with numerical bounds in Appendix Section~\ref{app:architecture}.

The direct branch predicts a separate waveform from phase, window progress, local decoder context, global state, and the cuff embedding. A learned cuff affine map adjusts its scale and offset. The physics branch, conditioned on $\mathbf z_w$ through FiLM \citep{perez2018film}, generates the full fields $A(t,x)$ and $U(t,x)$ and evaluates the arterial residuals. The tube law in Appendix~\ref{app:pde} yields $P_{\mathrm{tm}}$. The physics pressure is
\begin{equation}
    P_{\mathrm{phy}}(t)=P_{\mathrm{WK}}(t)
    +\eta\{P_{\mathrm{tm}}(t,0)-P_{\mathrm{tm}}(t,L)\},
\end{equation}
where $\eta=1.2\operatorname{sigmoid}(\eta_{\mathrm{raw}})$ is a learned dimensionless global scalar. The final pressure is
\begin{equation}
    P(t)=\alpha P_{\mathrm{dir}}(t)+(1-\alpha)P_{\mathrm{phy}}(t).
    \label{eq:pfinal}
\end{equation}

For a piecewise constant $U_L$ over $[t_i,t_{i+1}]$, the boundary ODE has the exact update
\begin{equation}
\begin{split}
    P_C(t_{i+1})={}e^{-\Delta t_i/\taumodel}P_C(t_i)+\left(1-e^{-\Delta t_i/\taumodel}\right)
    \left\{P_v+\taumodel U_L(t_i)\right\}.
\end{split}
\label{eq:wk_discrete}
\end{equation}
We integrate the boundary ODE with this exact update, avoiding the error introduced by explicit time stepping. The initial state is predicted from $\mathbf z_w$ as $P_C(t_0)=P_v+\operatorname{softplus}\{g(\mathbf z_w)\}+10^{-3}$. A periodic loss compares the first and last capacitive pressures over one learned cardiac cycle. We detach the Windkessel rollout from its parameter gradients, so the explicit $\taumodel$ objectives and priors train the window-level Windkessel parameters while the rollout constrains the field and pressure pathways.

The blend coefficient is a single global value shared across windows. Its selected value is reported in Appendix \ref{app:architecture}.

\subsection{Training Objectives and Curriculum}
\label{sec:training}

The phase stage minimizes
\begin{equation}
    \mathcal L_{\mathrm{phase}}=\mathcal L_{\mathrm{ref}}
    +\mathcal L_{\mathrm{peak}}+0.5\mathcal L_{\mathrm{mono}}
    +0.1\mathcal L_{\mathrm{smooth}}.
\end{equation}
The reference loss is offset invariant. A peak alignment term aligns ECG R peaks to multiples of $2\pi$. The last two terms penalize negative and rapidly changing phase increments.

The waveform stage uses a masked reconstruction objective combining absolute waveform, standardized shape, derivative, correlation, beat, and scalar pressure terms. Its exact composition and coefficients are given in Appendix Section~\ref{app:training_losses}.

Outside the phase stage, the training objective is
\begin{equation}
\begin{split}
\mathcal L=\mathcal L_{\mathrm{rec}}+\lambda_d\mathcal L_{\mathrm{dist}}
+\lambda_p\mathcal L_{\mathrm{phys}}+\lambda_r\mathcal L_{\mathrm{par}}+\lambda_b\mathcal L_{\mathrm{NIBP}}+\lambda_\tau\mathcal L_\tau
+\lambda_m\mathcal L_{\tau,\mathrm{morph}}+\lambda_v\mathcal L_{P_v}+\lambda_s\mathcal L_{S_U}.
\end{split}
\label{eq:total_loss}
\end{equation}
$\mathcal L_{\mathrm{dist}}$ is MSE to a frozen pressure teacher captured after the waveform stage. The cuff loss applies Smooth L1 loss to predicted SBP, DBP, and MAP after division by 20 mmHg. The direct time constant loss is a log-MSE against the proxy derived from pressure weighted by inverse variance, with weights $\mathrm{SE}(\log\tauwave)^{-2}$ on valid windows. It uses the same ABP-derived target evaluated in Section~\ref{sec:tau-consistency}, so that windows with more reliable decay fits contribute more to the objective. A second supervised term estimates log decay from the predicted blended pressure on diastolic segments and asymptotes selected from the true ABP, using a Smooth L1 loss with transition point $\beta=0.20$.

The physics group enforces Windkessel coupling, PDE residuals, periodicity, and parameter priors. Appendix Table~\ref{tab:supervision} lists every loss group with its weights and supervision source. We fix the compliance normalization to one and set the normalized distal resistance equal to $\taumodel$.

The curriculum has four stages over 140 epochs: phase learning, direct waveform learning, physics activation with a decreasing blend weight, and joint refinement. All reported results use a single trained model. The full training configuration is given in Appendix \ref{app:architecture} and Table~\ref{tab:supervision}.

Introducing the physics pathway against a frozen teacher preserves the waveform solution and is an optimization choice.

\section{Data and Evaluation Design}
\label{sec:data}

\subsection{Study Cohort}
% {\color{blue} VitalDB contains synchronized intraoperative ECG lead II, PPG, invasive ABP, NIBP, and clinical records from 2808 patients (2881 operations) sampled near 500 Hz. Recordings are divided into non-overlapping 10-second windows and resampled to 250 data points for the HIPNO model.

% To ensure data validity, we applied a retrospective five-step quality control (QC) filter evaluating signal coverage, waveform morphology (ECG, PPG, ABP), and ECG-ABP pulse timing. This process retained 945,499 high-quality windows from 2,562 patients, representing 25.2\% of the initial candidate pool.

% Patients were partitioned into five folds—stratified by eligible window count and label availability (NIBP, CVP, CO)—ensuring no data leakage across splits. This yielded a 60\%/20\%/20\% split for training (616,385 windows), validation (163,978), and testing (165,136). Specifically, the test set includes 59,239 ABP-evaluable windows (498 patients), \TauTestWindows{} $\tauwave$-eligible windows (\TauTestPatients{} patients), and \COTestWindows{} CO calibration windows (\COTestPatients{} patients); the latter also serves as the external reconstruction benchmark.

% Detailed preprocessing steps, QC rules, and comprehensive cohort statistics are provided in Appendix Section~\ref{app:waveform_preprocessing} and Tables~\ref{tab:qc_rules}--\ref{tab:task_cohorts}.}

We use the subset of VitalDB with synchronized intraoperative ECG lead II, PPG,
invasive ABP, NIBP, and clinical records from 2808 patients sampled near 500 Hz
\citep{lee2022vitaldb}. Recordings are divided into nonoverlapping 10s windows
and resampled to 250 points. Five rule groups then screen finite-sample
coverage, ECG rhythm, PPG morphology, ABP morphology, and pulse timing between
ECG and ABP (Appendix Table~\ref{tab:qc_rules}). The timing rule uses invasive
ABP, so this screen is retrospective. Quality control retains 1204591 windows from 2562 patients. Then, chronological stratified sampling caps each patient at 800 windows in training and 500 in
validation and test, giving the 945499 windows as model input (25.2\% of the
3746095 initial candidates).

Patients are partitioned into five folds stratified by eligible window count and
by NIBP, CVP, and CO label availability, so that all windows from one patient
remain in one split. Three folds form the training set and one fold each forms
the validation and test sets, giving 616385, 163978, and 165136 windows from
1535, 512, and 515 patients. The test partition provides 59239 ABP-evaluable
windows from 498 patients, \TauTestWindows{} $\tauwave$-eligible windows from
\TauTestPatients{} patients, and \COTestWindows{} CO calibration windows from
\COTestPatients{} patients, which also serve as the external reconstruction
benchmark. Appendix Section~\ref{app:waveform_preprocessing} and
Tables~\ref{tab:cohort} and~\ref{tab:task_cohorts} list the preprocessing steps
and full cohort counts.

\subsection{Reference Targets and Evaluation Cohorts}

The time constant reference is a proxy derived only from invasive ABP by identifying diastolic decay segments after the systolic downstroke, fitting the exponential decay model
\begin{equation}
    \log\{P(t)-P_\infty\}=a-t/\tau_{\mathrm{wave}}.
    \label{eq:tau_proxy}
\end{equation}
and aggregates the valid fits for individual beats within each window with weighting by inverse variance. Failed fits are marked invalid and are excluded from the time constant loss. Among windows with valid targets, median $\tauwave$ is 0.499 s, median fit $R^2$ is 0.969, and the median contributing beat count is nine. Appendix Section~\ref{app:tau_reference} gives the full fitting thresholds.

The CO cohort is formed from the optional monitor channels Vigileo, EV1000, Vigilance, CardioQ, and Solar8000 CO. Monitor CO is recorded approximately once per minute, whereas each analysis window lasts 10 s. Consequently, most clean windows contain no CO observation and only a minority contain one. No retained window contains multiple CO observations. A window is included in the CO cohort only when it contains one finite observation, which is used directly as the label with no averaging within the window. The label is matched to the frozen operator output by subject, record, and window, with values retained from 1 to 15 L/min.
Because these device outputs do not form a uniform independent standard, later calibration models include device indicators. Appendix Table~\ref{tab:co_devices} gives the complete device level counts before and after cleaning. Flotrac devices may span multiple product generations, whose arterial pressure waveform methods have evolved over time \citep{slagt2014systematic}.

The six cuff features described in Section~\ref{sec:formulation} are supplied to HIPNO as clinically available calibration variables. A conditioning MLP with widths 6, 128, and 16 maps them to a window level affine scale and shift for the direct pressure waveform. The affine head is initialized at the identity.

% \FloatBarrier
\subsection{Evaluation and Uncertainty}
\label{sec:evaluation}
Primary outcomes are held-out time constant prediction error (against $\tauwave$) and ABP reconstruction error, with patient clustered confidence intervals. The CO calibration and coordinate perturbations are exploratory. Confidence intervals are percentile intervals from patient-level resampling. The full cohort time constant analysis uses 5000 bootstrap iterations. The separate consistency analysis in 100 patients uses 2000, and ABP reconstruction uses 2000.

For the time constant, the population baseline predicts the median across the training cohort of the mean $\log\tauwave$ within each patient for every test window. For ABP, patient-level waveform and beat pressure metrics are averaged across patients. For CO agreement and perturbation uncertainty, 2000 bootstrap replicates resample patients with replacement and recompute each metric. Appendix Section~\ref{app:evaluation_details} defines the estimands and the clustered bootstrap protocol.

All model selection and calibration choices use the training and validation partitions only. The test set is evaluated once after those choices are fixed. Bootstrap resampling over patients preserves the dependence among windows contributed by one patient. Intervals are computed for one fitted checkpoint and quantify patient sampling variation. They do not include optimization variability across training seeds.

\begin{wrapfigure}{r}{0.55\columnwidth}
    \vspace{-50pt}
    \centering
    \includegraphics[width=\linewidth]{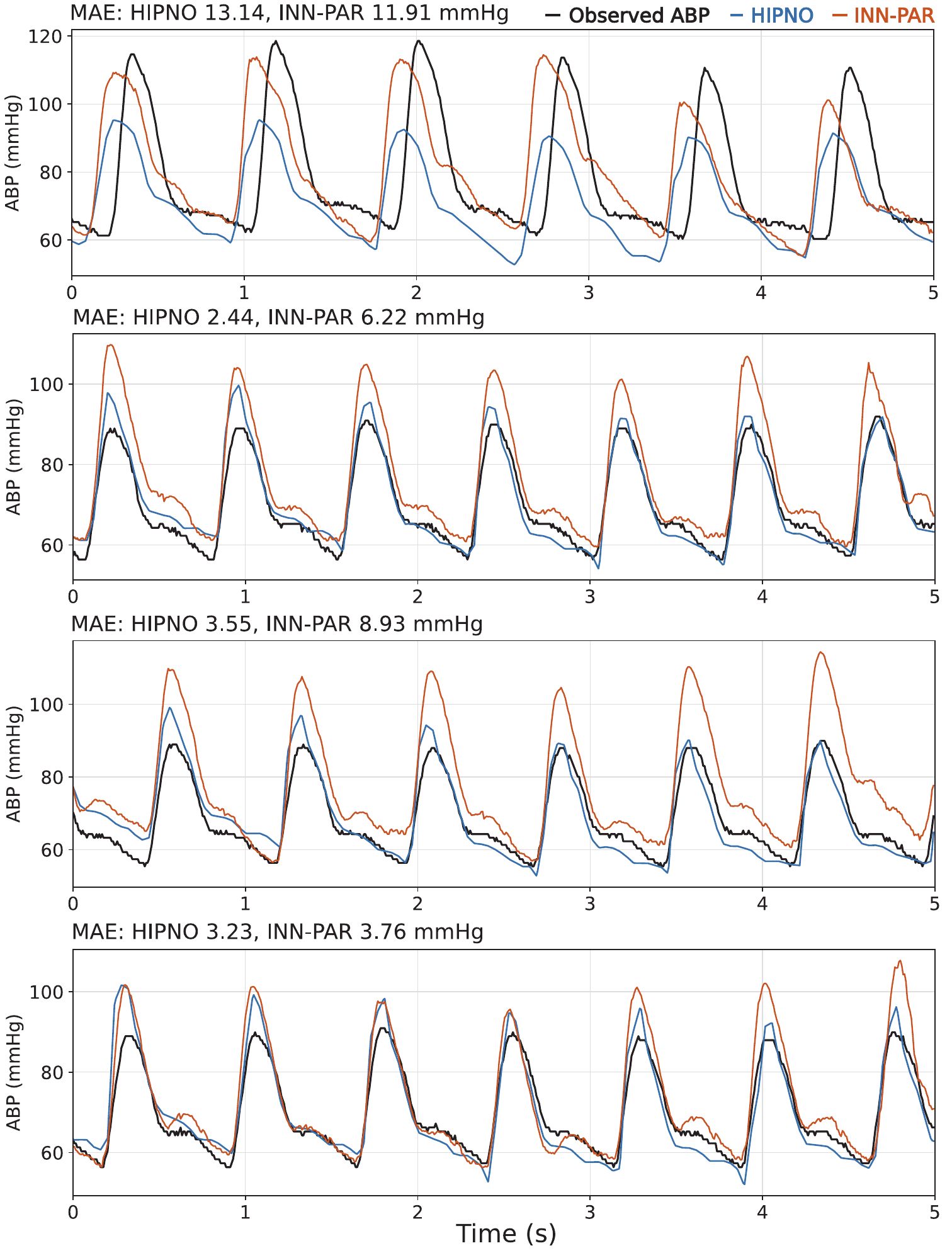}
    \caption{Illustrative 5s ABP excerpts from four held-out windows. Each panel
    overlays invasive ABP (black), HIPNO (blue), and INN-PAR (orange). Each
    title gives the per-window waveform MAE (mmHg).}
    \label{fig:waveform}
    \vspace{-50pt}
\end{wrapfigure}

\subsection{Comparison Protocol}
\label{sec:controlled_comparisons}

The evaluation is organized around the population baseline and external baseline methods. The training cohort median tests whether HIPNO predicts between-patient variation in the time constant beyond a population summary. For pressure reconstruction, we compare HIPNO with INN-PAR, a reversible waveform model that also uses signal gradients, whose invertible blocks jointly encode PPG and ABP \citep{kundu2025inn}, and PaPaGei, a foundation model for optical physiological signals whose PPG representations generalize across individuals \citep{pillai2410papagei}. These are the strongest reproducible baselines under our study split. The input modalities each method uses and their implementation details are given in Appendix Section~\ref{app:comparison_exclusions}.

The comparison is scored on pressure outputs. HIPNO's additional $\taumodel$, $U$, and $\kappa$ coordinates, monitor CO calibration, and coordinate perturbations are reported separately. BP-DeepONet-style physics-informed models lacked released code or checkpoints, while PPG2ABP and UniCardio exceeded the available compute budget. Both exclusions are documented in Appendix Section~\ref{app:comparison_exclusions}.

\section{Results}
\subsection{Held-Out Prediction of a Pressure-Derived Time Constant}
\label{sec:tau-consistency}
We compare HIPNO's predicted $\tauhat$ against the pressure-derived reference $\tauwave$ (Eq.~\ref{eq:tau_proxy}). Because the same proxy supervises training, this experiment characterizes prediction of $\tauwave$ in patients withheld from training. Resistance and compliance remain external physical quantities.

The test MAE in $\log\tauwave$ computed over windows falls from \TauBaselineTestMAE{} [0.408, 0.462] for the training median baseline to \TauTestMAE{} [0.276, 0.319], a 32\% relative reduction. RMSE falls from 0.525 [0.494, 0.555] to 0.382 [0.351, 0.414]. The predicted values $\tauhat$ capture structure that varies across patients with Pearson $r=\TauTestPearson{}$ [0.614, 0.734] and Spearman $\rho=0.674$ [0.615, 0.728]. Appendix Table~\ref{tab:tau_sensitivity} gives the corresponding reference range sensitivity analysis.

Estimates aggregated within patients are consistent with the result over windows. Among \TauTestPatients{} test patients, patient MAE is 0.237 s, absolute scale Pearson $r=0.672$, and logarithmic-scale Pearson $r=0.678$. Appendix Table~\ref{tab:tau_patient_agreement} gives the mean and agreement estimands, including the separate analysis in 100 patients with 2000 clustered bootstrap replicates.

Tightening the test reference interval to 0.32 to 2.20 s, 0.35 to 2.20 s, and 0.40 to 1.50s gives HIPNO MAE values of 0.287, 0.287, and 0.270. Correlation falls from 0.678 on the full range to 0.479 on the narrowest range as the variance of the target narrows.

\subsection{Mechanism Separation Responses in Hemodynamic Coordinates}
\label{sec:interventions}
\begin{wrapfigure}{r}{0.55\columnwidth}
    \vspace{-12pt}
    \centering
    \includegraphics[width=\linewidth]{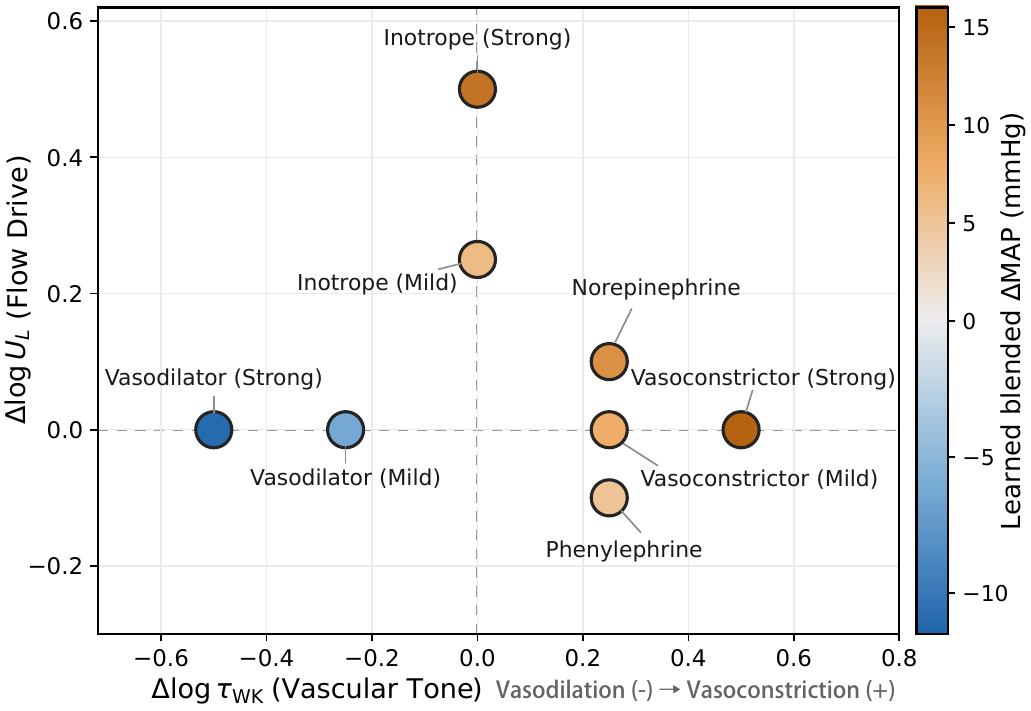}
    \caption{Mechanism separation in the quotient coordinates. Representative
    perturbations on the vascular tone ($\Delta\log\taumodel$) by flow drive
    ($\Delta\log U_L$) plane, colored by the learned blended $\Delta$MAP
    response from the Windkessel rollout.}
    \label{fig:mechanism_plane}
    \vspace{-25pt}
\end{wrapfigure}

The most direct quotient space test is mechanism separation, a capability no baseline restricted to pressure provides. In the context of hemodynamic management, interventions such as vasopressors and fluid boluses act on ABP-derived measurements by influencing vascular tone (see Appendix Table~\ref{tab:counterfactual_window}), suggesting directions we can assess for mechanistic consistency. Across the prespecified directions, the learned coordinates respond to changes in tone and flow drive with directional consistency of at least 90\% in all but two scenarios (full range 70.5-100\%). These counterfactuals do not estimate treatment effects or use medication exposure, dose, timing, or observed intervention response.

Each perturbation corresponds to a coordinate direction associated with a class of hemodynamic interventions in clinical practice: vasodilators, vasoconstrictive vasopressors, positive inotropes, preload manipulation, and agents of mixed mechanism such as norepinephrine and phenylephrine \citep{holmes2005vasoactive,overgaard2008inotropes,thiele2011clinical}. Appendix Table~\ref{tab:counterfactual_window} gives the full clinical class mapping and references.

We perturb $\taumodel$, $U_L$, and $P_v$ on the held-out CO test subset using 13 prespecified scenarios. The inputs, global representation, direct pressure branch, and excess pressure term remain fixed. For each scenario, we scale $U_L$ by $\exp(\Delta\log U)$, scale $\taumodel$ by $\exp(\Delta\log\taumodel)$, and add $\Delta P_v$, then rerun the Windkessel rollout after updating the derived resistance and capacitive pressure terms. We measure changes in the physics branch and in the final blended pressure.

For CO, each window's baseline prediction from the CO calibration model of Section~\ref{sec:co_validation} (Model~D) sets a fixed flow conversion, so $\Delta\mathrm{CO}=\widehat{\mathrm{CO}}_{D,\mathrm{base}}\{\exp(\Delta\log U)-1\}$. The fitted regression is not rerun and its pressure, $\tauhat$, demographic, and device features are not recomputed. Thus $\taumodel$ and $P_v$ perturbations leave CO fixed by construction, while changes in $U_L$ alter the frozen scale estimate. Because $\Delta$CO is imposed by this conversion, the informative outputs are the $\Delta$MAP and $\Delta$PP responses, which the learned flow field and Windkessel rollout produce rather than assume.

Directional consistency requires every prespecified output to match its expected nonzero sign. A zero response fails either sign. Appendix Table~\ref{tab:direction_rules} lists the scenario rules, and Table~\ref{tab:sensitivity_full} gives intervals from 2000 bootstrap replicates clustered by patient for the complete counterfactual summaries. The mapping to clinical classes and the pass rates across windows are reported separately in Appendix Table~\ref{tab:counterfactual_window}.

The responses are consistent with the imposed equations. Mean physics $\Delta$MAP is $-14.9$ and $-26.2$ mmHg under mild and strong vasodilation, versus $+18.2$ and $+37.1$ mmHg under mild and strong vasoconstriction. Vasoconstriction consistency is lower than vasodilation consistency (91.3\% versus 99.4 to 100\%). This asymmetry suggests that the learned pressure response is less constrained near the upper end of the $\taumodel$ range and warrants further investigation.

The mechanism separation test is most informative when the same tone change is paired with opposite flow drive changes (Figure~\ref{fig:mechanism_plane}), because this is the ambiguity clinicians face at the bedside. A fall in monitor CO can reflect either reduced flow drive (hemorrhage or cardiac failure) or increased vascular tone through compensatory constriction. The two have opposite treatment implications, calling for volume or inotropic support in the first case and afterload reduction in the second. A parameterization in the original $(\Qin,R_1,R_2,C)$ coordinates admits either interpretation along the pressure preserving scale family. In the quotient coordinates, with $\Delta\log\taumodel=+0.25$, raising $U_L$ increases CO by $+0.54$ L/min at 100\% directional consistency, while lowering $U_L$ decreases it by $-0.49$ L/min at 90.5\%. The change is assigned to a tone or a flow axis before an intervention is chosen.
Venous pressure responses follow algebraically from the imposed equations and the $\taumodel$ and $U_L$ responses propagate through the learned flow field.

\subsection{Pressure Reconstruction Cost and Branch Attribution}
\label{sec:reconstruction}

Having established that the coordinates respond as separate mechanistic axes, we next quantify the pressure reconstruction cost of additionally producing the quotient coordinates. HIPNO matches the mean pressure targets that most directly inform clinical decisions, and its waveform cost is limited in magnitude and concentrated in systolic peaks (Table~\ref{tab:reconstruction}).

For the matched external comparison, we evaluate HIPNO against PaPaGei \citep{pillai2410papagei} and INN-PAR \citep{kundu2025inn} on the benchmark subset shared with the CO analysis. Table~\ref{tab:reconstruction} contrasts the outputs each method produces with its pressure fidelity cost. HIPNO reconstructs pressure while also producing the quotient coordinates ($\taumodel$, $U$, $\kappa$) and the mechanism separation analysis, whereas INN-PAR is optimized for reversible waveform reconstruction and PaPaGei for pressure estimates from a pretrained foundation model. INN-PAR was retrained on the study training split. The public PaPaGei checkpoint was pretrained on VitalDB, so patient-level overlap with our test cohort may be present.

\begin{table*}[t]
    \centering

    % ==================== Left table ====================
    \begin{minipage}[t]{0.44\textwidth}
        \centering
        \small
        \setlength{\tabcolsep}{3.5pt}

        \captionof{table}{
        Capability and cost comparison on the held-out benchmark subset.
        Check marks denote capabilities a method provides.
        MAE is in mmHg (lower better) and Wave $r$ is waveform correlation
        (higher better).
        }
        \label{tab:reconstruction}

        \begin{tabular}{@{}lccccc@{}}
            \toprule
            & \multicolumn{2}{c}{Outputs \& analyses}
            & \multicolumn{3}{c}{Cost} \\
            \cmidrule(lr){2-3}\cmidrule(l){4-6}
            Method
            & \makecell{ABP\\wave}
            & \makecell{Mechanism\\separation}
            & \makecell{MAP\\MAE\,$\downarrow$}
            & \makecell{SBP\\MAE\,$\downarrow$}
            & \makecell{Wave\\$r$\,$\uparrow$} \\
            \midrule
            PaPaGei & scalar      & ---        & 9.01  & 12.38 & ---   \\
            INN-PAR & \checkmark  & ---        & 10.34 & 16.32 & 0.817 \\
            HIPNO   & \checkmark  & \checkmark & 10.88 & 20.15 & 0.672 \\
            \bottomrule
        \end{tabular}
    \end{minipage}
    \hfill
    % ==================== Right table ====================
    \begin{minipage}[t]{0.54\textwidth}
        \centering
        \small
        \setlength{\tabcolsep}{2pt}

        \captionof{table}{
        Monitor CO calibration ablation.
        Test MAE is in L/min and brackets give 95\% patient bootstrap intervals.
        }
        \label{tab:co_calibration}

        \textbf{Test error and association}\par\vspace{2pt}

        \begin{tabularx}{\linewidth}{@{}cXcc@{}}
            \toprule
            Model & Feature set & Test MAE & Pearson $r$ \\
            \midrule
            A & HR, MAP, PP, HIPNO $\tauhat$
              & \makecell{1.654\\{\scriptsize[1.393, 1.927]}}
              & \makecell{0.368\\{\scriptsize[0.183, 0.536]}} \\

            B & A plus EHR covariates
              & \makecell{1.117\\{\scriptsize[0.934, 1.296]}}
              & \makecell{0.777\\{\scriptsize[0.680, 0.853]}} \\

            C & B plus device
              & \makecell{1.096\\{\scriptsize[0.919, 1.270]}}
              & \makecell{0.788\\{\scriptsize[0.692, 0.862]}} \\

            D & C plus $\log\overline U_L$
              & \makecell{1.095\\{\scriptsize[0.914, 1.270]}}
              & \makecell{0.787\\{\scriptsize[0.693, 0.860]}} \\

            E & D features, patient median Huber fit
              & \makecell{1.059\\{\scriptsize[0.875, 1.239]}}
              & \makecell{0.796\\{\scriptsize[0.702, 0.868]}} \\
            \bottomrule
        \end{tabularx}
    \end{minipage}

\end{table*}

% \begin{table}[t]
% \centering
% \small
% \setlength{\tabcolsep}{3.5pt}
% \caption{Capability and cost comparison on the held-out benchmark subset. Check marks denote capabilities a method provides. MAE is in mmHg (lower better) and Wave $r$ is waveform correlation (higher better).}
% \label{tab:reconstruction}
% \begin{tabular}{@{}lccccc@{}}
%     \toprule
%     & \multicolumn{2}{c}{Outputs \& analyses} & \multicolumn{3}{c}{Cost} \\
%     \cmidrule(lr){2-3}\cmidrule(l){4-6}
%     Method & \makecell{ABP\\wave} & \makecell{Mechanism\\separation} & \makecell{MAP\\MAE\,$\downarrow$} & \makecell{SBP\\MAE\,$\downarrow$} & \makecell{Wave\\$r$\,$\uparrow$} \\
%     \midrule
%     PaPaGei & scalar     & ---        & 9.01  & 12.38 & --- \\
%     INN-PAR & \checkmark & ---        & 10.34 & 16.32 & 0.817 \\
%     HIPNO   & \checkmark & \checkmark & 10.88 & 20.15 & 0.672 \\
%     \bottomrule
% \end{tabular}
% \end{table}

The MAP cost is bounded. HIPNO stays within 0.54 mmHg of INN-PAR (10.88 versus 10.34 mmHg), while the waveform penalty concentrates in systolic morphology (SBP MAE 20.15 versus 16.32 mmHg, waveform correlation 0.672 versus 0.817). Exposing $\taumodel$ and $U$ enables the calibration against monitor CO and the tone/flow perturbations reported above.

Figure~\ref{fig:waveform} illustrates this phase-specific error. HIPNO rounds systolic peaks while preserving beat timing and diastolic recovery, whereas INN-PAR often overshoots peaks or tails. This explains why a larger waveform or SBP error can accompany a small MAP penalty, since localized peak deviations contribute less to a cycle average than a sustained level or tail offset.

Computing metrics for each branch in isolation makes the source of the blended output explicit. On the same benchmark, the direct branch has waveform MAE 21.86 mmHg and Pearson $r=0.699$. The physics branch has waveform MAE 24.61 mmHg and near-zero waveform correlation ($r=0.002$, Appendix Table~\ref{tab:branch_metrics_full}). Because the blend is $P=\alpha P_{\mathrm{dir}}+(1-\alpha)P_{\mathrm{phy}}$, the $P_{\mathrm{phy}}$ column is also the $\alpha=0$ output. The physics branch enforces the Windkessel and PDE constraints and shapes the shared representation. Its standalone $r\approx0$ is by design, and the blend is the evaluated output.

\subsection{Monitored Cardiac Output Calibration}
\label{sec:co_validation}

Following the supervised calibration protocol in Section~\ref{sec:data}, the analysis joins monitor CO to frozen operator outputs one to one by subject, record, and clean window order. We fit nested log-linear calibration models using predicted MAP, predicted pulse pressure, ECG-derived heart rate, HIPNO's $\tauhat$, age, sex, BSA, device source, and mean distal latent flow. Model E replaces the ridge fit weighted across windows with a Huber fit to the median within each patient. Appendix Section~\ref{app:co_protocol} gives the feature transformations, candidate grids, weighting, cross validation, and the details of refits in which the target device is held out.

% \begin{table}[t]
%     \centering
%     \small
%     \setlength{\tabcolsep}{2pt}
%     \caption{Monitor CO calibration ablation. Test MAE is in L/min and brackets give 95\% patient bootstrap intervals.}
%     \label{tab:co_calibration}
%     \textbf{Test error and association}\par\vspace{2pt}
%     \begin{tabularx}{\columnwidth}{@{}cXcc@{}}
%         \toprule
%         Model & Feature set & Test MAE & Pearson $r$ \\
%         \midrule
%         A & HR, MAP, PP, HIPNO $\tauhat$ & \makecell{1.654\\{\scriptsize[1.393, 1.927]}} & \makecell{0.368\\{\scriptsize[0.183, 0.536]}} \\
%         B & A plus EHR covariates & \makecell{1.117\\{\scriptsize[0.934, 1.296]}} & \makecell{0.777\\{\scriptsize[0.680, 0.853]}} \\
%         C & B plus device & \makecell{1.096\\{\scriptsize[0.919, 1.270]}} & \makecell{0.788\\{\scriptsize[0.692, 0.862]}} \\
%         D & C plus $\log\overline U_L$ & \makecell{1.095\\{\scriptsize[0.914, 1.270]}} & \makecell{0.787\\{\scriptsize[0.693, 0.860]}} \\
%         E & D features, patient median Huber fit & \makecell{1.059\\{\scriptsize[0.875, 1.239]}} & \makecell{0.796\\{\scriptsize[0.702, 0.868]}} \\
%         \bottomrule
%     \end{tabularx}

% \end{table}

Adding EHR covariates accounts for essentially all of the improvement. Test MAE falls by 0.537 L/min from Model A (1.654) to Model B (1.117), whereas adding device source and HIPNO's flow coordinate $\log\overline U_L$ each change MAE by at most 0.021 L/min (the $\log\overline U_L$ term moves MAE by 0.001). Model E reaches validation MAE 0.762 and test MAE \COTestMAE{} L/min with test Pearson $r=0.796$ and Spearman $\rho=0.771$. No test set information was used for model selection.

HIPNO's coordinates can be included in the calibration model without degrading it, but do not by themselves improve upon the clinical covariates. Thus, the experiment tests the feasibility of monitor CO calibration. A static model using demographics alone is not the relevant comparator in the intraoperative setting, where anesthesia, ventilation, positioning, blood loss, and vasoactive drugs move flow far from any resting anthropometric estimate \citep{king2023physiology,teboul2016less}.

The Bland-Altman statistics defined in Section~\ref{sec:evaluation} are summarized for Model E in the text and given for both nested calibration models in Appendix Table~\ref{tab:co_calibration_agreement}. Model E has bias $-0.358$ L/min and percentage error \COTestPE{}, which exceeds the conventional 30\% agreement criterion for comparison with a reference CO method \citep{critchley1999meta}. The delay-aware transfer table is the primary device-specific analysis.

We additionally evaluated transfer using a forward matching window chosen separately for each device from each monitor's update mechanism. CardioQ's esophageal Doppler signal updates on every beat, so a 1s window allows only minimal hardware and integration latency \citep{singer1993esophageal}. EV1000/FloTrac integrates the pulse contour over a rolling window of 20s \citep{saugel2021cardiac}. Vigilance continuous thermodilution was evaluated with 3 min (STAT) and 7 min (standard) windows because stochastic heating and cross-correlation produce a delayed update \citep{yelderman1990continuous,siegel1996delayed}.

Under this protocol, Model E has PE 57.5\% for CardioQ, 44.5\% on the EV1000 test split, and 30.6\% and 27.9\% for Vigilance at 3 and 7 min, respectively. For CardioQ and Vigilance, the target device was excluded from calibration model fitting and its train, validation, and test observations were pooled into combined evaluation cohorts (21 and 3 patients, respectively). The complete transfer table (Appendix Table~\ref{tab:delay_device_transfer}) and its column definitions are in Appendix Section~\ref{app:delay_matching}.

\section{Discussion}
\label{sec:discussion}

% \textbf{Clinical significance.} Our results suggest that HIPNO can enable continuous, non-invasive hemodynamic monitoring. While non-invasive ABP estimation from surrogate waveforms has been extensively explored, HIPNO leverages the known coupling between pressure and flow to noninvasively and continuously track cardiac output and vascular tone, which has previously remained out of reach. We validate these findings in three contexts. 
% First, it maintains high-fidelity mean arterial pressure (MAP) reconstruction (Table~\ref{tab:reconstruction}), ensuring agreement with the clinical gold standard linked to perioperative acute kidney injury, myocardial injury, and mortality \citep{sessler2019poqi}. Second, HIPNO continuously tracks the vascular decay constant ($\taumodel$) from routine non-invasive waveforms, providing a real-time index of arterial compliance and systemic vascular tone \citep{bikia2024novel,westerhof2009arterial} that previously required invasive arterial lines. Third, by decoupling cardiac flow drive ($U$) from vascular tone, HIPNO provides the mechanistic insight necessary to distinguish hypovolemia from vasodilation, directly informing targeted fluid administration versus vasopressor titration during acute hemodynamic instability (Section~\ref{sec:interventions}).

\textbf{Clinical significance.} Noninvasive arterial pressure reconstruction is
well studied \citep{ibtehaz2022ppg2abp,kundu2025inn,chen2026versatile}, but
cardiac flow drive and vascular tone still require catheterization to observe.
HIPNO recovers both from ECG and PPG by using the known physical coupling
between pressure and flow. The decay constant $\taumodel$ provides a continuous
index of arterial compliance and vascular tone
\citep{bikia2024novel,westerhof2009arterial} that otherwise requires an arterial
line. Separating $\taumodel$ from flow drive $U$ distinguishes hypovolemia from
vasodilation, the distinction that determines whether volume or a vasopressor is
indicated (Section~\ref{sec:interventions}). HIPNO retains mean arterial
pressure (MAP) accuracy (Table~\ref{tab:reconstruction}), the quantity most
directly linked to perioperative acute kidney injury, myocardial injury, and
mortality \citep{sessler2019poqi}.

Compared to PaPaGei and INN-PAR, which target the pressure waveform alone, HIPNO also
recovers the invariant coordinates $\taumodel$, $U$, and $\kappa$ and makes the
residual scale ambiguity explicit. More generally, mechanistic symmetries should
be quotiented out before latent states are learned. The construction of
Section~\ref{sec:methods} instantiates this principle in the four steps below.

\begin{center}
\fbox{\parbox{\dimexpr\columnwidth-2\fboxsep-2\fboxrule\relax}{%
\textbf{Symmetry-aware operator design}\par\vspace{2pt}\hrule\vspace{4pt}
(1)~Write the observation map and identify its symmetry group $G$. For the Windkessel model, Proposition~\ref{prop:scale}
gives the pressure-preserving scaling of $(\Qin,R_1,R_2,C)$ (Sections~\ref{sec:formulation} and~\ref{sec:identifiability}).\\[3pt]
(2)~Form invariant coordinates that quotient out $G$. Here $U=Q/C$, $\taumodel=R_2C$, and $\kappa=R_1C$ (Section~\ref{sec:identifiability}).\\[3pt]
(3)~Parameterize and train the operator in those coordinates with predictive and residual losses. Here a physics-conditioned trunk with staged supervision (Sections~\ref{sec:architecture} and~\ref{sec:training}).\\[3pt]
(4)~Identify the external measurement that fixes the remaining scale convention. Here a compliance or flow reference (Measurement prescription, below).%
}}
\end{center}

Steps (1) to (3) follow directly once $G$ is known, and step (4) is problem
specific. The same construction applies beyond hemodynamics. In compartmental
pharmacokinetics, a concentration-time course determines the elimination rate,
which is the ratio of clearance to volume of distribution and serves as the
invariant coordinate. In subsurface Darcy flow, permeability and viscosity enter
the governing equations only through their ratio, so mobility is the coordinate
recoverable from an inverse problem in pressure and flow \citep{Oliver_Reynolds_Liu_2008}.

% \textbf{Transfer under matched update delays.} Matching each monitor to its own update mechanism exposes behavior specific to each device that a uniform window would blur. The larger EV1000 cohort transfers best, while the small CardioQ and Vigilance cohorts remain exploratory (Section~\ref{sec:co_validation}, Appendix Table~\ref{tab:delay_device_transfer}).
% % 

\textbf{Measurement prescription.} Device-specific transfer analyses demonstrate that establishing an absolute scale requires both an external reference and temporal alignment with its acquisition mechanism. Without this alignment, monitor latency and temporal integration risk being conflated with model error. While a compliance reference establishes the scale mapping normalized inputs to volumetric flow and dual resistances, an independent flow reference calibrates this scale directly as illustrated by the monitor CO analysis in Section~\ref{sec:co_validation}. Supported by the largest cohort, the EV1000 results yield the most robust evidence for transferability, whereas CardioQ and Vigilance estimates remain preliminary (Appendix Table~\ref{tab:delay_device_transfer}). Echocardiographic stroke volume serves as an additional candidate reference. Although anthropometrics can provide a fixed reference point for this calibration, tracking dynamic intraoperative changes requires the high-frequency signals HIPNO already uses. Therefore, future prospective studies should pair noninvasive inputs with temporally matched scale references and clinically meaningful intervention labels.

\section{Limitations and Ethical Considerations}
Two limitations constrain what these results establish. First, the evidence is
retrospective. VitalDB comprises surgical patients at a single tertiary care
center selected by acuity and device availability. Window eligibility and the
reference labels both derive from invasive signals that the target deployment
settings lack. Generalization to less acute populations, unrepresented cohorts,
and other sensor hardware therefore requires external prospective evaluation
with recalibration before deployment. Second, the 3-element Windkessel model
captures core arterial dynamics but omits venous return, ventricular
interaction, distributed wave reflections, and autonomic regulation.

Deployment as an uncalibrated diagnostic in high-acuity care risks clinical
misinterpretation, such as misjudged fluid responsiveness. Translation will
therefore require signal-quality-based rejection, uncertainty thresholds that
gate clinical action, and failure analysis in the high-risk states where
inference matters most (arrhythmia, low perfusion, vasoactive support).
Evaluation should extend beyond agreement with monitor labels to actionable
outcomes such as acute kidney injury, myocardial injury, and mortality. Pending such evidence, the model is intended for research use only.

VitalDB acquisition and public release were approved by the Institutional Review
Board of Seoul National University Hospital (H-1408-101-605), with consent
waived for retrospective release of the anonymized records. These data are
distributed under CC-BY 4.0 \citep{lee2022vitaldb}. We used them within those
terms and made no attempt to re-identify patients.

\section{Generative AI Usage}

Generative AI tools were used for language editing and coding assistance. They were not used to generate experimental data or for any other purposes.

\section{Data and Code Availability}
VitalDB is publicly available under its stated access terms \citep{lee2022vitaldb}. Our code is released in the Github repository shown in the first page. Other artifacts will be released upon acceptance of this paper.

\section{Conclusion}

HIPNO demonstrates how a physics-informed neural operator can learn in the identifiable quotient space of a mechanistic observation model. In hemodynamics, that choice yields prediction of between-patient variation in the decay constant, functionally separated tone and flow drive responses, and a limited pressure cost while leaving the missing absolute scale explicit. A compliance or independent flow reference is the concrete next measurement needed to turn these invariant coordinates into absolute variables. Clinically, HIPNO recovers vascular tone and flow drive from ECG and PPG alone, a step toward mechanistic hemodynamic insight for patients who never receive an invasive line.

%% Example review annotations (removed with [final] option)
%%   \alice{comment}       — plain comment
%%   \alice[r]{text}       — strikethrough (suggest removal)
%%   \alice[h]{text}       — highlight (draw attention)
%\alice{This paragraph needs more detail on the experimental setup.}
%\alice[r]{This sentence is redundant and should be removed.}
%\alice[h]{This claim needs a citation.}
%\bob{Looks good to me.}

\bibliographystyle{plainnat}
\bibliography{references}

%% ============================================================
%% Appendix
%% ============================================================
%% Start appendix with centered title and table of contents
%% Default title: "Supplementary Materials"; customize with \makeappendixtitle[Your Title]
\makeappendixtitle

\appendix
\clearpage
\onecolumn

\section{Windkessel and Arterial Dynamics}
\label{app:pde}

Our arterial model uses a normalized compliant segment $x\in[0,1]$ and the compliance-normalized flow $U$. Its residual equations are
\begin{equation}
    \frac{\partial A}{\partial t}+\frac{\partial U}{\partial x}=0,
\end{equation}
\begin{equation}
    \frac{\partial U}{\partial t}
    +\frac{\partial}{\partial x}\left(\alpha_{\mathrm c}\frac{U^2}{A}\right)
    +c_p A\frac{\partial P_{\mathrm{tm}}}{\partial x}
    +C_f\frac{U}{A}=0,
\end{equation}
where $\alpha_{\mathrm c}=1.1$ is fixed and $c_p>0$ and $C_f>0$ are learned coefficients. This is the residual our formulation enforces. It is not an absolute flow equation with measured vessel geometry. The tube law is
\begin{equation}
    P_{\mathrm{tm}}(t,x)=P_{\mathrm{ext}}
    +\beta\left(\sqrt{A/A_0}-1\right),
    \label{eq:tube_law}
\end{equation}
where $A_0=1$ is fixed and $P_{\mathrm{ext}}$ and $\beta>0$ are learned \citep{sherwin2003computational}. In the trained model, the learned external pressure settles near 70 mmHg (a physiological transmural offset). The remaining coefficients are positive and stable ($\beta\approx3.9$, $c_p\approx1.70$, $C_f\approx2.8\times10^{-3}$, normalized). The field output uses 16 uniform spatial locations. The residual uses 24 phase locations and six uniform spatial locations. Centered periodic differences approximate time derivatives. Centered interior and one sided boundary differences approximate spatial derivatives. Huber losses exclude the two spatial boundary nodes.

At the inlet, $U(t,0)$ is fixed to a smooth positive pulse whose onset, width, sharpness, shoulder, and floor are predicted from the window representation. The field output is shifted so this condition holds exactly. At the outlet, the Windkessel update in Eq.~\ref{eq:wk_discrete} supplies the pressure constraint. The operator has no separate initial condition for $A$ or $U$ because it is evaluated on a periodic phase representation. The learned $P_C(0)$ and the periodic capacitive pressure loss close the window level Windkessel rollout.

\section{Architecture and Training Environment}
\label{app:architecture}

For each modality $m\in\{\mathrm{ECG},\mathrm{PPG}\}$, the two convolutional layers have 128 channels, kernels 15 and 3, and unit stride. Two pre-norm Transformer layers use model dimension 128, feed forward dimension 256, eight attention heads, GELU, and zero dropout. The two modalities do not share weights. The decoder uses independent eight head ECG and PPG cross-attention modules. Its local fusion maps 384 to 256 to 64 dimensions. Its global fusion maps 512 to 256 to 64 dimensions.

The trunk input has dimension 37 and maps to width 256. It receives 12 phase harmonics and six spatial Fourier frequencies, but not normalized progress or local decoder context. Each of eight residual blocks applies layer normalization, a 256 to 1024 to 256 MLP, GELU, and no dropout. FiLM uses
\begin{equation}
    \operatorname{FiLM}(h,\mathbf z_w)=\{1+\gamma(\mathbf z_w)\}\odot h+\delta(\mathbf z_w),
\end{equation}
with independent linear maps for $\gamma$ and $\delta$, both initialized at zero. FiLM is applied before the residual blocks and after final layer normalization. The model has 5.70M trainable parameters.

The direct branch receives 10 phase harmonics, normalized progress, the local decoder context, $\mathbf z_w$, and the cuff embedding. Its MLP widths are 256, 256, and 128, and its output bias is initialized at 80 mmHg. The learned cuff affine map then adjusts scale and offset. The PDE residual is evaluated on 24 phase points and six spatial points using centered differences, periodic differences in phase, and interior residuals in space.

The Windkessel head maps the 64-dimensional global representation through a 128-unit GELU layer to $\kappa$, $\taumodel$, and $P_v$. Compliance is fixed to one by convention. $\kappa$ is bounded to $[0.02,0.15]$, $\taumodel$ to $[0.30,2.50]$ s, and $P_v$ to $[2,20]$ mmHg. The $P_v$ soft prior is centered at 5 mmHg with scale 3 mmHg and maximum weight 0.10. The global pressure correction coefficient $\eta$ is regularized toward 0.25 and has selected value 0.254. The exact ODE update, phase loss, loss weights, and curriculum are given in Sections~\ref{sec:architecture} and \ref{sec:training}. 

Cuff features are standardized as $(x-120)/40$ for SBP, $(x-70)/25$ for DBP, $(x-90)/30$ for MAP, and $(x-50)/30$ for pulse pressure, and the cuff time offset is divided by 600 s. We train with AdamW at learning rate $10^{-4}$, weight decay $10^{-5}$, and gradient norm clipped at one, using a 3000-step linear warmup followed by cosine decay to $10^{-6}$. The batch size is 16 and training runs for 140 epochs at 9000 steps per epoch, or 1260000 scheduled steps, in bfloat16 mixed precision with an exponential moving average decay of 0.995. The four-stage curriculum covers phase learning in epochs 1 to 8, direct waveform learning in epochs 9 to 60, physics activation with a decreasing $\alpha$ in epochs 61 to 120, and joint refinement in epochs 121 to 140. The selected checkpoint blends the two branches with the global weights $\alpha=0.5763$ (direct) and $1-\alpha=0.4237$ (physics). With random seed 42, we select epoch 124, which attains the lowest validation $\log\tau_{\mathrm{wave}}$ mean absolute error, 0.3347 over the 8000 validation windows evaluated each epoch. All reported HIPNO results use the exponential moving average weights at that epoch. Training uses two GPUs with PyTorch DistributedDataParallel and takes approximately 109 hours of wall time.

\section{Comparison Scope}
\label{app:comparison_exclusions}

The two included baselines provide complementary points of comparison for pressure prediction. INN-PAR represents reversible and gradient aware waveform reconstruction, whereas PaPaGei represents cross-individual foundation model pretraining for optical signals. External baselines were included only when an official implementation or checkpoint could be run under the study split. INN-PAR was trained with its released pipeline after resampling PPG and ABP to 125 Hz, using nonoverlapping 5s windows, PPG and its gradient as inputs, Adam at $10^{-4}$, batch size 128, 501 epochs, and seed 1. PaPaGei uses 10s PPG segments at 125 Hz, a frozen 512-dimensional encoder (5.79M parameters), separate ridge heads for SBP, DBP, and MAP, and the alpha grid $\{0.1,1,10,100,1000\}$ with four fold selection. BP-DeepONet-like models lacked released code or checkpoints, so their weights and hyperparameters could not be reproduced reliably. The compute required for larger systems such as PPG2ABP and UniCardio exceeded one month per retraining run, beyond the feasible evaluation budget. Thus, these methods were not included in the empirical comparison.

\begin{table}[H]
    \centering
    \small
    \setlength{\tabcolsep}{3pt}
    \caption{Complete branch isolated pressure metrics on the held-out benchmark. Errors are in mmHg.}
    \label{tab:branch_metrics_full}
    \begin{tabularx}{\textwidth}{@{}l>{\centering\arraybackslash}X>{\centering\arraybackslash}X@{}}
        \toprule
        Metric & $P_{\mathrm{dir}}$ (direct) & $P_{\mathrm{phy}}$ (physics) \\
        \midrule
        Waveform MAE & 21.86 & 24.61 \\
        Waveform RMSE & 26.43 & 30.76 \\
        Waveform Pearson $r$ & 0.699 & 0.002 \\
        SBP MAE & 40.16 & 43.86 \\
        DBP MAE & 11.88 & 28.22 \\
        MAP MAE & 17.55 & 20.36 \\
        SBP Pearson $r$ & 0.660 & 0.275 \\
        DBP Pearson $r$ & 0.640 & $-0.178$ \\
        MAP Pearson $r$ & 0.713 & $-0.145$ \\
        \bottomrule
    \end{tabularx}
\end{table}

\begin{table}[H]
    \centering
    \small
    \setlength{\tabcolsep}{4pt}
    \caption{Bland-Altman agreement for nested CO calibration models. Bias and limits of agreement (LoA) in L/min. Brackets give 95\% patient-clustered bootstrap intervals.}
    \label{tab:co_calibration_agreement}
    \begin{tabularx}{\textwidth}{@{}l>{\centering\arraybackslash}X>{\centering\arraybackslash}X>{\centering\arraybackslash}X>{\centering\arraybackslash}X@{}}
        \toprule
        Model & Bias & Lower LoA & Upper LoA & PE \\
        \midrule
        D
          & \makecell{$-0.428$\\{\scriptsize[$-0.673$, $-0.188$]}}
          & \makecell{$-2.964$\\{\scriptsize[$-3.456$, $-2.402$]}}
          & \makecell{2.109\\{\scriptsize[1.858, 2.325]}}
          & \makecell{47.8\%\\{\scriptsize[41.5, 53.2]}} \\
        E
          & \makecell{$-0.358$\\{\scriptsize[$-0.592$, $-0.133$]}}
          & \makecell{$-2.874$\\{\scriptsize[$-3.338$, $-2.351$]}}
          & \makecell{2.158\\{\scriptsize[1.907, 2.383]}}
          & \makecell{\COTestPE\\{\scriptsize[41.2, 52.2]}} \\
        \bottomrule
    \end{tabularx}
\end{table}

\section{Estimands and Bootstrap Uncertainty}
\label{app:evaluation_details}

For the window level time constant, point estimates pool all valid windows without patient reweighting, and percentile intervals resample patients. For the complementary patient aggregated estimand, patient $p$ with $n_p$ valid windows has within patient means $\overline{\widehat\tau}_p=n_p^{-1}\sum_i\widehat\tau_{pi}$ and $\overline\tau_p=n_p^{-1}\sum_i\tau_{pi}$, where $\widehat\tau_{pi}$ is the prediction $\tauhat$ and $\tau_{pi}$ the target $\tauwave$. Each patient then receives equal weight. Mean predicted $\tauhat$ and observed $\tauwave$ are averaged across patients, patient MAE and RMSE summarize $\overline{\widehat\tau}_p-\overline\tau_p$, log-MAE summarizes $|\log\overline{\widehat\tau}_p-\log\overline\tau_p|$, and Pearson $r$ and Spearman $\rho$ are computed across patient-level pairs. All patient aggregated time constant intervals use this patient clustered bootstrap. Patient identifiers are sampled with replacement and all eligible windows from each selected patient are retained with that patient's multiplicity. The full cohort tables use 5000 replicates. The dedicated 100-patient consistency analysis uses 2000 replicates.

For ABP reconstruction, all valid waveform samples for one patient are pooled to compute that patient's MAE, RMSE, and correlation. Beat level SBP, DBP, and MAP errors are pooled within patient. The point estimate is the mean of the resulting patient metrics.

For CO agreement and perturbation uncertainty, each bootstrap replicate samples $P$ patient identifiers with replacement, where $P$ is the number of patients in the evaluated split. A selected patient contributes all eligible windows, repeated according to that patient's bootstrap multiplicity. Bland-Altman bias, lower and upper limits of agreement, and bias-corrected percentage error are recomputed from this clustered sample. Counterfactual replicates similarly recompute mean physics and blended MAP and pulse pressure changes, mean CO change, output specific direction rates, and consistency across all prespecified directions. Appendix Table~\ref{tab:direction_rules} lists the directional consistency rules. The 95\% interval is the empirical 2.5th to 97.5th percentile of the replicate values.

\section{Cardiac Output Calibration Procedure}
\label{app:co_protocol}

Let $\mathbf z_i$ contain the feature set selected for calibration model $i$. Models A to D fit
\begin{equation}
    \log \widehat{\mathrm{CO}}_i=\beta_0+\boldsymbol\beta^\top\mathbf z_i,
    \quad
    \widehat{\boldsymbol\beta}
    =\arg\min_{\boldsymbol\beta}\sum_j w_j e_j^2
    +\lambda\lVert\boldsymbol\beta\rVert_2^2.
\end{equation}
Here $e_j=\log\mathrm{CO}_j-(\beta_0+\boldsymbol\beta^\top\mathbf z_j)$ is the residual on training window $j$. All positive continuous features other than age and sex are included on the log scale. Features are standardized with training means and standard deviations. Constant features are dropped, and rows with missing features are excluded without imputation. The weights $w_j$ give every training patient equal total weight regardless of window count. Ridge candidates are $\lambda\in\{0, 0.001, 0.01, 0.1, 1, 10, 100\}$, selected by five fold subject grouped cross validation on patient-level MAE.

Model E first takes the median of every feature and the CO reference within each training patient. It fits a penalized Huber regression on log CO with $\lambda_H\in\{0.001, 0.01, 0.1, 1, 10, 100\}$ and $\epsilon_H\in\{1.35, 1.5, 2.0\}$, selected by patient-level MAE. At evaluation, a prediction is produced for each window, and patient-level metrics compare the within patient median prediction with the within patient median reference. Device categories are learned from training data, one category is omitted as reference, and unseen test categories cause an error. For refits on held-out target devices, the target device is designated as the omitted reference category, so no target specific coefficient is estimated from absent target device examples.

\section{Training Losses and Supervision}
\label{app:training_losses}

The waveform reconstruction loss is
\begin{equation}
\begin{split}
\mathcal L_{\mathrm{rec}}={}\mathcal L_{\mathrm{abs}}+0.25\mathcal L_z
+2\mathcal L_{\partial t}+0.5\mathcal L_{\mathrm{corr}}+0.5\mathcal L_{\mathrm{beat}}+6.5\mathcal L_{\mathrm{stat}}.
\end{split}
\label{eq:reconstruction_loss}
\end{equation}
$\mathcal L_{\mathrm{abs}}$ is masked MSE in mmHg. $\mathcal L_z$ is masked Huber loss after standardization by the true window mean and standard deviation. $\mathcal L_{\partial t}$ is masked derivative MSE. The remaining terms match correlation, beat level SBP/DBP, and window level SBP/DBP/MAP. Soft extrema use temperature 2 mmHg. For the final blend, branch gradients scale as
\begin{equation}
    \frac{\partial\mathcal L}{\partial P_{\mathrm{dir}}}
    =\alpha\frac{\partial\mathcal L}{\partial P},\quad
    \frac{\partial\mathcal L}{\partial P_{\mathrm{phy}}}
    =(1-\alpha)\frac{\partial\mathcal L}{\partial P}.
\end{equation}

\begin{table}[H]
    \centering
    \small
    \setlength{\tabcolsep}{4.5pt}
    \caption{Loss groups in the final trained model. Inner weights multiply terms before the scheduled group weights.}
    \label{tab:supervision}
    \begin{tabularx}{\textwidth}{@{}p{0.95in}Xp{0.9in}p{1.55in}@{}}
        \toprule
        Group & Terms and inner weights & \makecell{Maximum\\group weight} & Source \\
        \midrule
        Reconstruction & Eq.~\ref{eq:reconstruction_loss} & 1.00 & Invasive ABP \\
        Distillation & Frozen waveform teacher & 2.00 & Model prediction \\
        Physics & Coupling 10, WK 0.1, mass 0.1, PDE 0.05, total resistance 0.5, periodicity 0.1, flow smoothness 0.12, latent smoothness 0.05, proximal pressure 0.25 & 0.20 & Equations and ABP \\
        Parameters & Prior 0.75, within subject consistency 0.05 & 0.05 & Bounds and priors \\
        Cuff & Soft SBP, DBP, MAP & 0.03 & NIBP \\
        Time constant & Direct log MSE 1.0, morphology 0.5 & 1.00, 0.50 & ABP-derived $\tau_{\mathrm{wave}}$ \\
        Scale conventions & $P_v$ soft prior, log $S_U$ soft prior & 0.10 each & Fixed prior centers \\
        \bottomrule
    \end{tabularx}
\end{table}

\section{Waveform Preprocessing}
\label{app:waveform_preprocessing}

We store each waveform as 5000 samples over a nonoverlapping 10s window, corresponding to the nominal 500 Hz rate and 2 ms sampling interval used for reference construction. Before model input, linear interpolation maps each window to 250 points while preserving its duration. The output mask is interpolated and retained only where its value is at least 0.999. The preprocessing does not impute failed waveform windows. Consecutive identical finite NIBP rows are treated as one cuff event and are not interpolated. NIBP may be missing because its validity indicator is part of the input. No train fitted waveform normalization is applied.

\section{Cohort and Device Counts}

\begin{table}[H]
    \centering
    \small
    \setlength{\tabcolsep}{4.5pt}
    \caption{Evaluation cohort sizes by task and split.}
    \label{tab:task_cohorts}
    \begin{tabular}{@{}llrr@{}}
        \toprule
        Task & Split & Windows & Patients \\
        \midrule
        \multirow{3}{*}{ABP reconstruction}
          & Training & 225058 & 1478 \\
          & Validation & 61974 & 495 \\
          & Test & 59239 & 498 \\
        \midrule
        \multirow{3}{*}{Time constant}
          & Training & 122398 & 1427 \\
          & Validation & \TauValidationWindows{} & \TauValidationPatients{} \\
          & Test & \TauTestWindows{} & \TauTestPatients{} \\
        \midrule
        \multirow{3}{*}{CO calibration}
          & Training & 59286 & 361 \\
          & Validation & 16321 & 122 \\
          & Test & \COTestWindows{} & \COTestPatients{} \\
        \bottomrule
    \end{tabular}
\end{table}

\begin{table}[H]
    \centering
    \small
    \setlength{\tabcolsep}{5pt}
    \caption{Waveform quality control rules.}
    \label{tab:qc_rules}
    \begin{tabularx}{\textwidth}{@{}p{1.25in}X@{}}
        \toprule
        Rule group & Applied criteria \\
        \midrule
        Finite samples and variation & ECG, PPG, and ABP require at least 98\% finite samples and at least 32 distinct values. \\
        ECG rhythm & Reject large robust range, discontinuous steps, baseline jumps, fewer than seven detected beats in 10 s, heart rate outside 45 to 180 beats/min, or RR coefficient of variation above 0.35. \\
        PPG morphology & Require measurable spread, at least seven peaks, and interbeat coefficient of variation at most 0.45. \\
        ABP morphology & Require MAP from 50 to 160 mmHg, adjacent steps no greater than 30 mmHg, limited flat segments, at least seven peaks and troughs, and interbeat coefficient of variation at most 0.40. Matched cycles require SBP from 70 to 220 mmHg, DBP from 30 to 140 mmHg, pulse pressure from 15 to 130 mmHg, and pulse pressure coefficient of variation at most 0.65. \\
        Pulse timing & Match ECG R peaks one to one to later ABP systolic peaks with lag from 0.05 to 0.45 s. At least 75\% must match, lag coefficient of variation must be at most 0.35, and PPG and matched ECG beat counts may differ by at most one. \\
        \bottomrule
    \end{tabularx}
\end{table}

  Cohort construction applies two filters to the quality-controlled windows.
Waveform quality control (Table~\ref{tab:qc_rules}) retains 1204591 of the
3746095 candidate windows. Patients are assigned to training, validation, and
test with no patient crossing splits. Within each patient, chronological
stratified sampling caps the contribution at 800 windows in training and 500 in
validation and test, which gives the 945499 windows for model input (616385 training,
163978 validation, 165136 test). Supervision targets are matched separately. A
window is NIBP-supervised when a cuff measurement falls within 300s of its
midpoint. Of the model input windows, 346271 carry such a match, and
NIBP supervision is applied to the 263482 matched training windows and
skipped on the rest.

\begin{table}[H]
    \centering
    \small
    \setlength{\tabcolsep}{2.5pt}
    \caption{Complete coordinate perturbation summaries. Brackets give 95\% percentile intervals from 2000 patient clustered bootstrap replicates. Pressure changes are in mmHg, CO changes are in L/min, and consistency is the fraction satisfying every prespecified sign. Blend columns hold the direct branch fixed.}
    \label{tab:sensitivity_full}
    \begin{tabularx}{\textwidth}{@{}p{1.25in}>{\centering\arraybackslash}X>{\centering\arraybackslash}X>{\centering\arraybackslash}X>{\centering\arraybackslash}X>{\centering\arraybackslash}X>{\centering\arraybackslash}X@{}}
    \toprule
    Scenario & Physics $\Delta$MAP & Blend $\Delta$MAP & Physics $\Delta$PP & Blend $\Delta$PP & $\Delta$CO & Consistency \\
    \midrule
    Vasodilator-like, mild
      & \makecell{$-14.918$\\{\scriptsize[$-15.258$, $-14.562$]}}
      & \makecell{$-6.321$\\{\scriptsize[$-6.464$, $-6.170$]}}
      & \makecell{11.992\\{\scriptsize[10.343, 13.634]}}
      & \makecell{4.037\\{\scriptsize[3.797, 4.275]}}
      & \makecell{0.000\\{\scriptsize[0.000, 0.000]}}
      & \makecell{99.4\%\\{\scriptsize[98.8, 99.9]}} \\
    Vasoconstrictor-like, mild
      & \makecell{18.152\\{\scriptsize[17.457, 18.874]}}
      & \makecell{7.691\\{\scriptsize[7.396, 7.997]}}
      & \makecell{$-13.411$\\{\scriptsize[$-16.252$, $-10.570$]}}
      & \makecell{$-4.425$\\{\scriptsize[$-4.889$, $-3.976$]}}
      & \makecell{0.000\\{\scriptsize[0.000, 0.000]}}
      & \makecell{91.3\%\\{\scriptsize[87.9, 94.4]}} \\
    Inotrope-like, mild
      & \makecell{14.364\\{\scriptsize[14.138, 14.585]}}
      & \makecell{6.086\\{\scriptsize[5.990, 6.180]}}
      & \makecell{4.075\\{\scriptsize[2.848, 5.175]}}
      & \makecell{$-0.964$\\{\scriptsize[$-1.419$, $-0.496$]}}
      & \makecell{1.453\\{\scriptsize[1.384, 1.519]}}
      & \makecell{70.5\%\\{\scriptsize[65.0, 75.5]}} \\
    Vasodilator-like, strong
      & \makecell{$-26.213$\\{\scriptsize[$-26.644$, $-25.752$]}}
      & \makecell{$-11.106$\\{\scriptsize[$-11.289$, $-10.911$]}}
      & \makecell{20.083\\{\scriptsize[17.155, 22.959]}}
      & \makecell{7.187\\{\scriptsize[6.791, 7.568]}}
      & \makecell{0.000\\{\scriptsize[0.000, 0.000]}}
      & \makecell{100.0\%\\{\scriptsize[100.0, 100.0]}} \\
    Vasoconstrictor-like, strong
      & \makecell{37.132\\{\scriptsize[35.946, 38.342]}}
      & \makecell{15.732\\{\scriptsize[15.230, 16.245]}}
      & \makecell{$-8.439$\\{\scriptsize[$-13.841$, $-3.243$]}}
      & \makecell{$-7.057$\\{\scriptsize[$-8.090$, $-5.965$]}}
      & \makecell{0.000\\{\scriptsize[0.000, 0.000]}}
      & \makecell{91.3\%\\{\scriptsize[87.9, 94.4]}} \\
    Inotrope-like, strong
      & \makecell{32.807\\{\scriptsize[32.291, 33.313]}}
      & \makecell{13.900\\{\scriptsize[13.681, 14.114]}}
      & \makecell{21.959\\{\scriptsize[20.197, 23.442]}}
      & \makecell{$-0.187$\\{\scriptsize[$-1.304$, 0.933]}}
      & \makecell{3.318\\{\scriptsize[3.161, 3.470]}}
      & \makecell{95.8\%\\{\scriptsize[93.0, 97.7]}} \\
    Venous pressure down
      & \makecell{$-2.687$\\{\scriptsize[$-2.746$, $-2.628$]}}
      & \makecell{$-1.138$\\{\scriptsize[$-1.164$, $-1.114$]}}
      & \makecell{$-1.469$\\{\scriptsize[$-1.694$, $-1.252$]}}
      & \makecell{$-0.343$\\{\scriptsize[$-0.411$, $-0.274$]}}
      & \makecell{0.000\\{\scriptsize[0.000, 0.000]}}
      & \makecell{100.0\%\\{\scriptsize[100.0, 100.0]}} \\
    Venous pressure up
      & \makecell{2.687\\{\scriptsize[2.628, 2.746]}}
      & \makecell{1.138\\{\scriptsize[1.114, 1.164]}}
      & \makecell{1.469\\{\scriptsize[1.252, 1.694]}}
      & \makecell{0.368\\{\scriptsize[0.298, 0.439]}}
      & \makecell{0.000\\{\scriptsize[0.000, 0.000]}}
      & \makecell{100.0\%\\{\scriptsize[100.0, 100.0]}} \\
    Phenylephrine-like
      & \makecell{11.760\\{\scriptsize[11.066, 12.471]}}
      & \makecell{4.983\\{\scriptsize[4.688, 5.284]}}
      & \makecell{$-16.380$\\{\scriptsize[$-18.824$, $-13.985$]}}
      & \makecell{$-3.409$\\{\scriptsize[$-3.893$, $-2.951$]}}
      & \makecell{$-0.487$\\{\scriptsize[$-0.509$, $-0.464$]}}
      & \makecell{90.5\%\\{\scriptsize[86.9, 93.9]}} \\
    Norepinephrine-like
      & \makecell{25.216\\{\scriptsize[24.506, 25.952]}}
      & \makecell{10.684\\{\scriptsize[10.383, 10.996]}}
      & \makecell{$-7.790$\\{\scriptsize[$-10.919$, $-4.623$]}}
      & \makecell{$-5.287$\\{\scriptsize[$-5.754$, $-4.806$]}}
      & \makecell{0.538\\{\scriptsize[0.513, 0.563]}}
      & \makecell{100.0\%\\{\scriptsize[100.0, 100.0]}} \\
    Dobutamine-like
      & \makecell{6.279\\{\scriptsize[5.863, 6.695]}}
      & \makecell{2.660\\{\scriptsize[2.484, 2.836]}}
      & \makecell{4.422\\{\scriptsize[3.575, 5.293]}}
      & \makecell{0.999\\{\scriptsize[0.528, 1.474]}}
      & \makecell{1.453\\{\scriptsize[1.384, 1.519]}}
      & \makecell{76.2\%\\{\scriptsize[71.3, 80.7]}} \\
    Fluid-bolus-like
      & \makecell{17.051\\{\scriptsize[16.771, 17.324]}}
      & \makecell{7.224\\{\scriptsize[7.106, 7.340]}}
      & \makecell{5.534\\{\scriptsize[4.253, 6.662]}}
      & \makecell{$-0.514$\\{\scriptsize[$-1.046$, 0.027]}}
      & \makecell{1.453\\{\scriptsize[1.384, 1.519]}}
      & \makecell{100.0\%\\{\scriptsize[100.0, 100.0]}} \\
    Preload-reduction-like
      & \makecell{$-13.873$\\{\scriptsize[$-14.098$, $-13.645$]}}
      & \makecell{$-5.878$\\{\scriptsize[$-5.973$, $-5.781$]}}
      & \makecell{3.609\\{\scriptsize[2.472, 4.693]}}
      & \makecell{1.577\\{\scriptsize[1.232, 1.927]}}
      & \makecell{$-1.131$\\{\scriptsize[$-1.183$, $-1.078$]}}
      & \makecell{100.0\%\\{\scriptsize[100.0, 100.0]}} \\
    \bottomrule
    \end{tabularx}
\end{table}

\begin{table}[H]
    \centering
    \small
    \setlength{\tabcolsep}{6pt}
    \caption{Split level cohort counts. The initial candidate total includes unassigned data. Later stages contain the three assigned splits only.}
    \label{tab:cohort}
    \begin{tabular}{@{}llrr@{}}
        \toprule
        Stage & Split & Windows & Patients \\
        \midrule
        \multirow{4}{*}{\makecell[l]{Before cleaning\\candidate windows}}
          & Training & 2135497 & 1535 \\
          & Validation & 703389 & 512 \\
          & Test & 710879 & 515 \\
          & All, including unassigned & 3746095 & 2808 \\
        \midrule
        \multirow{4}{*}{\makecell[l]{After waveform\\quality control}}
          & Training & 718961 & 1535 \\
          & Validation & 243007 & 512 \\
          & Test & 242623 & 515 \\
          & All & 1204591 & 2562 \\
        \midrule
        \multirow{4}{*}{\makecell[l]{After quality control\\and label-validity filtering}}
          & Training & 263482 & 1478 \\
          & Validation & 93383 & 495 \\
          & Test & 87884 & 498 \\
          & All & 444749 & 2471 \\
        \midrule
        \multirow{4}{*}{\makecell[l]{Final split selected\\model input}}
          & Training & 616385 & 1535 \\
          & Validation & 163978 & 512 \\
          & Test & 165136 & 515 \\
          & All & 945499 & 2562 \\
        \bottomrule
    \end{tabular}
\end{table}

\begin{table}[H]
    \centering
    \small
    \setlength{\tabcolsep}{5pt}
    \caption{Cardiac output labels by device and split before cleaning and in the final calibrated cohort. Device-specific rows may overlap. Totals are de-duplicated within each split. Solar8000 contributed no eligible CO labels.}
    \label{tab:co_devices}
    \begin{tabular}{@{}llrrrr@{}}
        \toprule
        & & \multicolumn{2}{c}{Before cleaning} & \multicolumn{2}{c}{Final CO calibrated} \\
        \cmidrule(lr){3-4}\cmidrule(l){5-6}
        Split & Device & Windows & Patients & Windows & Patients \\
        \midrule
        \multirow{5}{*}{Training}
          & Vigileo & 148990 & 127 & 15709 & 109 \\
          & EV1000 & 360096 & 258 & 42112 & 241 \\
          & Vigilance & 4993 & 5 & 249 & 3 \\
          & CardioQ & 24264 & 19 & 2410 & 16 \\
          & All devices union & 525952 & 400 & 59286 & 361 \\
        \midrule
        \multirow{4}{*}{Validation}
          & Vigileo & 36457 & 36 & 3501 & 30 \\
          & EV1000 & 142329 & 97 & 12479 & 91 \\
          & CardioQ & 3820 & 2 & 554 & 2 \\
          & All devices union & 180084 & 134 & 16321 & 122 \\
        \midrule
        \multirow{4}{*}{Test}
          & Vigileo & 54658 & 49 & 3584 & 42 \\
          & EV1000 & 131526 & 84 & 11826 & 77 \\
          & CardioQ & 4065 & 3 & 460 & 3 \\
          & All devices union & 186794 & 134 & \COTestWindows{} & \COTestPatients{} \\
        \bottomrule
    \end{tabular}
\end{table}

\section{Reference Time Constant Construction}
\label{app:tau_reference}

For each matched ABP cycle, the reference procedure locates the first negative peak after systole with prominence at least 1 mmHg or 3\% of pulse pressure. The decay segment starts 0.03s after this notch and ends 0.05s before the next trough. Segments shorter than 0.18s or with a pressure drop below 3 mmHg are rejected. For each candidate asymptote $P_\infty$ on a 64 point grid from 0 mmHg to 1 mmHg below the segment minimum, we fit Eq.~\ref{eq:tau_proxy}. The fit with the highest $R^2$ is retained if $R^2\geq0.70$ and $0.30\leq\tau_{\mathrm{wave}}\leq2.50$ s.

A window requires at least three valid beats. Its reference is the inverse variance weighted mean of beat level $\log\tauwave$, exponentiated for the second scale value. The stored standard error is the larger of the inverse information estimate and the between beat standard error. The reference fit uses 1.60M beat level decay segments.

\section{Reference Range Sensitivity}

Table~\ref{tab:tau_sensitivity} lists the values used to contextualize the main result. On the test set, HIPNO $\log\tauwave$ MAE is 0.287 in both intermediate ranges and 0.270 in the narrowest range. Pearson correlation declines as the reference range narrows, which is expected under range restriction.

\begin{table}[H]
    \centering
    \small
    \setlength{\tabcolsep}{3.5pt}
    \caption{Sensitivity of $\log\tauwave$ prediction to the accepted pressure-derived reference range. Brackets give 95\% patient bootstrap intervals. Counts are windows.}
    \label{tab:tau_sensitivity}
    \begin{tabular}{@{}llccccc@{}}
        \toprule
        Split & $\tauwave$ range in s & $n$ & Baseline MAE & HIPNO MAE & Pearson $r$ & Spearman $\rho$ \\
        \midrule
        \midrule
        Test & all & 34869 & \makecell{0.434\\{\scriptsize[0.408, 0.462]}} & \makecell{0.297\\{\scriptsize[0.276, 0.319]}} & \makecell{0.678\\{\scriptsize[0.614, 0.734]}} & \makecell{0.674\\{\scriptsize[0.615, 0.728]}} \\
        Test & $\tauwave\geq0.32$ & 29667 & \makecell{0.410\\{\scriptsize[0.383, 0.438]}} & \makecell{0.289\\{\scriptsize[0.267, 0.313]}} & \makecell{0.656\\{\scriptsize[0.588, 0.718]}} & \makecell{0.647\\{\scriptsize[0.580, 0.707]}} \\
        Test & $0.32\leq\tauwave\leq2.20$ & 28924 & \makecell{0.395\\{\scriptsize[0.369, 0.421]}} & \makecell{0.287\\{\scriptsize[0.265, 0.309]}} & \makecell{0.623\\{\scriptsize[0.554, 0.691]}} & \makecell{0.621\\{\scriptsize[0.553, 0.685]}} \\
        Test & $0.35\leq\tauwave\leq2.20$ & 25736 & \makecell{0.379\\{\scriptsize[0.352, 0.407]}} & \makecell{0.287\\{\scriptsize[0.264, 0.310]}} & \makecell{0.597\\{\scriptsize[0.518, 0.667]}} & \makecell{0.593\\{\scriptsize[0.518, 0.664]}} \\
        Test & $0.40\leq\tauwave\leq1.50$ & 18852 & \makecell{0.298\\{\scriptsize[0.277, 0.319]}} & \makecell{0.270\\{\scriptsize[0.249, 0.291]}} & \makecell{0.479\\{\scriptsize[0.392, 0.561]}} & \makecell{0.496\\{\scriptsize[0.411, 0.575]}} \\
        \bottomrule
    \end{tabular}
\end{table}

\section{Patient-Level Time Constant Agreement}
\label{app:tau_patient}

The patient-level test estimand first averages the predicted $\tauhat$ and reference $\tauwave$
within each patient and then computes the summary across patients, so every
patient receives equal weight. The test cohort contains $\TauTestPatients{}$
patients. For each interval, these patient identifiers are resampled with
replacement 5000 times, retaining all eligible windows for every selected
patient. The displayed limits are the 2.5th and 97.5th percentiles.

\begin{table}[H]
    \centering
    \small
    \setlength{\tabcolsep}{3.5pt}
    \caption{Patient level agreement between $\tauhat$ and $\tauwave$. Within patient means are formed before computing each metric. The full test rows use 5000 patient clustered bootstrap replicates. The 100-patient rows use 2000.}
    \label{tab:tau_patient_agreement}
    \textbf{Panel A (patient means)}\par\vspace{2pt}
    \begin{tabular}{@{}lrrccccccc@{}}
        \toprule
        Split & \makecell{Valid\\windows} & Patients & \makecell{Mean predicted\\$\tauhat$, s} & \makecell{Mean observed\\$\tauwave$, s} & \makecell{Patient\\MAE} & \makecell{Patient\\RMSE} & Log-MAE & $r$ & $\rho$ \\
        \midrule
        Validation & \TauValidationWindows{} & \TauValidationPatients{} & \makecell{0.705\\{\scriptsize[0.680, 0.733]}} & \makecell{0.794\\{\scriptsize[0.756, 0.834]}} & \makecell{0.233\\{\scriptsize[0.212, 0.254]}} & \makecell{0.324\\{\scriptsize[0.296, 0.353]}} & \makecell{0.284\\{\scriptsize[0.264, 0.306]}} & \makecell{0.699\\{\scriptsize[0.645, 0.746]}} & \makecell{0.690\\{\scriptsize[0.632, 0.741]}} \\
        Test & \TauTestWindows{} & \TauTestPatients{} & \makecell{0.668\\{\scriptsize[0.644, 0.694]}} & \makecell{0.764\\{\scriptsize[0.722, 0.804]}} & \makecell{0.237\\{\scriptsize[0.215, 0.260]}} & \makecell{0.342\\{\scriptsize[0.308, 0.377]}} & \makecell{0.297\\{\scriptsize[0.275, 0.319]}} & \makecell{0.672\\{\scriptsize[0.598, 0.739]}} & \makecell{0.682\\{\scriptsize[0.624, 0.736]}} \\
        \bottomrule
    \end{tabular}\par
    \vspace{5pt}
    \textbf{Panel B (agreement metrics)}\par\vspace{2pt}
    \begin{tabularx}{\textwidth}{@{}l>{\centering\arraybackslash}X>{\centering\arraybackslash}X>{\centering\arraybackslash}X>{\centering\arraybackslash}X>{\centering\arraybackslash}X@{}}
        \toprule
        Metric space & MAE & RMSE & Bias & Pearson $r$ & Spearman $\rho$ \\
        \midrule
        Absolute scale (s)
          & \makecell{0.237\\{\scriptsize[0.215, 0.260]}}
          & \makecell{0.342\\{\scriptsize[0.308, 0.377]}}
          & \makecell{$-0.095$\\{\scriptsize[$-0.127$, $-0.068$]}}
          & \makecell{0.672\\{\scriptsize[0.598, 0.739]}}
          & \makecell{0.682\\{\scriptsize[0.624, 0.736]}} \\
        Logarithmic scale
          & \makecell{0.297\\{\scriptsize[0.276, 0.319]}}
          & \makecell{0.382\\{\scriptsize[0.351, 0.412]}}
          & \makecell{$-0.039$\\{\scriptsize[$-0.072$, $-0.004$]}}
          & \makecell{0.678\\{\scriptsize[0.612, 0.733]}}
          & \makecell{0.674\\{\scriptsize[0.610, 0.728]}} \\
        \midrule
        100-patient absolute scale (s)
          & \makecell{0.281\\{\scriptsize[0.260, 0.303]}}
          & \makecell{0.415\\{\scriptsize[0.388, 0.445]}}
          & \makecell{$-0.094$\\{\scriptsize[$-0.129$, $-0.064$]}}
          & \makecell{0.647\\{\scriptsize[0.577, 0.708]}}
          & \makecell{0.626\\{\scriptsize[0.575, 0.673]}} \\
        100-patient logarithmic scale
          & \makecell{0.364\\{\scriptsize[0.344, 0.385]}}
          & \makecell{0.476\\{\scriptsize[0.451, 0.504]}}
          & \makecell{$-0.044$\\{\scriptsize[$-0.085$, $-0.002$]}}
          & \makecell{0.650\\{\scriptsize[0.593, 0.699]}}
          & \makecell{0.626\\{\scriptsize[0.574, 0.671]}} \\
        \bottomrule
    \end{tabularx}
\end{table}

\section{Reconstruction Stratified by Time Since Cuff Measurement}
\label{app:cuff_age_reconstruction}

Reconstruction error was stable across bins of time since the last NIBP measurement on the matched CO test subset ($\COTestWindows{}$ windows, $\COTestPatients{}$ patients, every window has a valid cuff match within the 300s limit). Across the 0-60, 61-180, and 181-300s bins, blended MAP MAE ranged 7.4-7.8 mmHg and blended waveform MAE was $\approx$12.2 mmHg.

\section{Device Matching and Agreement with Delay Accounted}
\label{app:delay_matching}

The device transfer analysis reports both a deployment like contemporaneous match and a delay-aware forward match. A contemporaneous match pairs the prediction at time $t_i$ one to one with the record from the same subject, record, device, and timestamp. For a device with an update delay $\delta$, the forward oracle match searches only after the prediction time and within the allowed physical window.
\begin{equation}
 j^*(i,\delta)=\underset{j:\,0\leq t_j-t_i\leq\delta}{\arg\min}\;\left|\widehat y_i-y_j\right|,
 \quad y_i^{\mathrm{oracle}}=y_{j^*(i,\delta)}.
 \label{eq:forward_oracle}
\end{equation}
Windows without an eligible record are omitted. The oracle match is an evaluation diagnostic because it uses the realized prediction error to select a reference. It is not a deployment time predictor. We use $\delta=1$s for CardioQ, $20$s for EV1000/FloTrac, and $180$s or $420$s for Vigilance. These values reflect the instruments' update mechanisms: beat-to-beat esophageal Doppler for CardioQ \citep{singer1993esophageal}, rolling pulse contour integration for FloTrac \citep{manecke2005edwards,saugel2021cardiac}, and stochastic heating/cross-correlation with slow and STAT update modes for Vigilance \citep{yelderman1990continuous,siegel1996delayed}.

For a small target device, a combined cohort pools that device's eligible observations from training, validation, and test after removing every target device patient from calibration model fitting. Thus CardioQ contributes a 21-patient combined evaluation cohort and Vigilance a 3-patient cohort, while the EV1000 rows retain their split labels. A patient clustered bootstrap samples patient identifiers with replacement and retains all eligible windows from each selected patient with its multiplicity. We use 2000 replicates for the rows accounting for delay. Repeated measures variance decomposes the residual variance into within patient variation and variation of patient average residuals, $s_{\mathrm{RM}}^2=s_{\mathrm{within}}^2+s_{\mathrm{between}}^2$, before computing percentage error. For prediction $\widehat y$ and reference $y$, the reported columns are
\begin{equation}
\begin{split}
 \mathrm{Bias}&=\overline{\widehat y-y},\quad
 \mathrm{LoA}_{\pm}=\mathrm{Bias}\pm1.96s_{\mathrm{RM}},\\
 \mathrm{PE}&=100\,\frac{1.96s_{\mathrm{RM}}}{(\overline{\widehat y}+\overline y)/2},\quad
 \mathrm{MAPE}=\frac{100}{n}\sum_{i=1}^{n}\frac{|\widehat y_i-y_i|}{|y_i|}.
\end{split}
 \label{eq:co_agreement_columns}
\end{equation}
The table reports MAE, RMSE, bias corrected PE, and Pearson correlation. All intervals are the 2.5th to 97.5th percentiles of the clustered bootstrap replicates.

\begin{table}[H]
    \centering
    \small
    \setlength{\tabcolsep}{3pt}
    \renewcommand{\arraystretch}{1.08}
    \caption{Delay-aware target device transfer for independently refit Models D and E (forward windows in the row labels). Parentheses give 95\% patient cluster bootstrap intervals. PE uses the repeated measures residual standard deviation.}
    \label{tab:delay_device_transfer}
    \begin{tabularx}{\textwidth}{@{}p{1.60in}c>{\centering\arraybackslash}X>{\centering\arraybackslash}X>{\centering\arraybackslash}X>{\centering\arraybackslash}X>{\centering\arraybackslash}X@{}}
        \toprule
        \makecell[l]{Device cohort} & Model & MAE & RMSE & Bias & PE & Pearson $r$ \\
        \midrule
        CardioQ (1s delay) & D
          & \makecell{1.07\\{\scriptsize[0.85, 1.34]}}
          & \makecell{1.39\\{\scriptsize[1.12, 1.69]}}
          & \makecell{$-0.00$\\{\scriptsize[$-0.44$, 0.43]}}
          & \makecell{57.0\%\\{\scriptsize[47.2\%, 64.3\%]}}
          & \makecell{0.66\\{\scriptsize[0.39, 0.78]}} \\
        CardioQ (1s delay) & E
          & \makecell{1.08\\{\scriptsize[0.87, 1.34]}}
          & \makecell{1.41\\{\scriptsize[1.15, 1.70]}}
          & \makecell{0.03\\{\scriptsize[$-0.40$, 0.47]}}
          & \makecell{57.5\%\\{\scriptsize[47.9\%, 65.0\%]}}
          & \makecell{0.64\\{\scriptsize[0.36, 0.76]}} \\
        \midrule
        EV1000 training (20 s) & D
          & \makecell{0.88\\{\scriptsize[0.79, 0.98]}}
          & \makecell{1.21\\{\scriptsize[1.08, 1.33]}}
          & \makecell{$-0.22$\\{\scriptsize[$-0.36$, $-0.07$]}}
          & \makecell{45.5\%\\{\scriptsize[41.6\%, 49.1\%]}}
          & \makecell{0.71\\{\scriptsize[0.66, 0.76]}} \\
        EV1000 training (20 s) & E
          & \makecell{0.88\\{\scriptsize[0.79, 0.99]}}
          & \makecell{1.23\\{\scriptsize[1.08, 1.40]}}
          & \makecell{$-0.14$\\{\scriptsize[$-0.29$, 0.01]}}
          & \makecell{46.6\%\\{\scriptsize[41.4\%, 51.8\%]}}
          & \makecell{0.70\\{\scriptsize[0.65, 0.76]}} \\
        EV1000 validation (20 s) & D
          & \makecell{0.92\\{\scriptsize[0.79, 1.06]}}
          & \makecell{1.24\\{\scriptsize[1.07, 1.42]}}
          & \makecell{$-0.12$\\{\scriptsize[$-0.35$, 0.12]}}
          & \makecell{48.3\%\\{\scriptsize[42.6\%, 53.2\%]}}
          & \makecell{0.68\\{\scriptsize[0.59, 0.76]}} \\
        EV1000 validation (20 s) & E
          & \makecell{0.93\\{\scriptsize[0.81, 1.07]}}
          & \makecell{1.25\\{\scriptsize[1.09, 1.42]}}
          & \makecell{$-0.06$\\{\scriptsize[$-0.29$, 0.18]}}
          & \makecell{48.6\%\\{\scriptsize[43.3\%, 53.3\%]}}
          & \makecell{0.67\\{\scriptsize[0.58, 0.75]}} \\
        EV1000 test (20 s) & D
          & \makecell{1.01\\{\scriptsize[0.83, 1.21]}}
          & \makecell{1.37\\{\scriptsize[1.13, 1.59]}}
          & \makecell{$-0.49$\\{\scriptsize[$-0.76$, $-0.22$]}}
          & \makecell{45.1\%\\{\scriptsize[39.1\%, 50.3\%]}}
          & \makecell{0.76\\{\scriptsize[0.67, 0.84]}} \\
        EV1000 test (20 s) & E
          & \makecell{0.98\\{\scriptsize[0.81, 1.16]}}
          & \makecell{1.33\\{\scriptsize[1.10, 1.55]}}
          & \makecell{$-0.40$\\{\scriptsize[$-0.67$, $-0.15$]}}
          & \makecell{44.5\%\\{\scriptsize[38.3\%, 50.0\%]}}
          & \makecell{0.76\\{\scriptsize[0.66, 0.84]}} \\
        \midrule
        Vigilance (3 min) & D
          & \makecell{0.41\\{\scriptsize[0.19, 1.31]}}
          & \makecell{0.59\\{\scriptsize[0.25, 1.32]}}
          & \makecell{$-0.16$\\{\scriptsize[$-0.59$, 1.31]}}
          & \makecell{30.1\%\\{\scriptsize[11.8\%, 44.8\%]}}
          & \makecell{0.68\\{\scriptsize[$-0.59$, 0.80]}} \\
        Vigilance (3 min) & E
          & \makecell{0.43\\{\scriptsize[0.22, 1.36]}}
          & \makecell{0.61\\{\scriptsize[0.29, 1.38]}}
          & \makecell{$-0.19$\\{\scriptsize[$-0.58$, 1.36]}}
          & \makecell{30.6\%\\{\scriptsize[14.6\%, 45.4\%]}}
          & \makecell{0.68\\{\scriptsize[$-0.53$, 0.80]}} \\
        Vigilance (7 min) & D
          & \makecell{0.33\\{\scriptsize[0.12, 1.31]}}
          & \makecell{0.52\\{\scriptsize[0.19, 1.32]}}
          & \makecell{$-0.10$\\{\scriptsize[$-0.48$, 1.31]}}
          & \makecell{27.2\%\\{\scriptsize[11.0\%, 41.8\%]}}
          & \makecell{0.71\\{\scriptsize[$-0.49$, 0.84]}} \\
        Vigilance (7 min) & E
          & \makecell{0.35\\{\scriptsize[0.15, 1.36]}}
          & \makecell{0.53\\{\scriptsize[0.22, 1.38]}}
          & \makecell{$-0.12$\\{\scriptsize[$-0.46$, 1.36]}}
          & \makecell{27.9\%\\{\scriptsize[12.3\%, 42.6\%]}}
          & \makecell{0.71\\{\scriptsize[$-0.41$, 0.84]}} \\
        \bottomrule
    \end{tabularx}
\end{table}

\section{Complete Coordinate Perturbation Results}

\begin{table}[H]
    \centering
    \small
    \setlength{\tabcolsep}{5pt}
    \caption{Expected response directions for coordinate perturbations. A zero change fails a required nonzero sign.}
    \label{tab:direction_rules}
    \begin{tabularx}{\textwidth}{@{}p{2.0in}X@{}}
    \toprule
    Scenario & Required signs \\
    \midrule
    Vasodilator-like and vasoconstrictor-like
      & MAP, DBP, and $\log\taumodel$ follow the sign of $\Delta\log\taumodel$. \\
    Inotrope-like
      & MAP, pulse pressure, and CO increase. \\
    Venous pressure down/up
      & MAP, DBP, and $P_v$ follow the sign of $\Delta P_v$. \\
    Phenylephrine-like
      & MAP, DBP, and $\log\taumodel$ increase. CO decreases. \\
    Norepinephrine-like
      & MAP, DBP, CO, and $\log\taumodel$ increase. \\
    Dobutamine-like
      & Pulse pressure and CO increase. $\log\taumodel$ decreases. \\
    Fluid-bolus-like and preload-reduction-like
      & MAP, CO, and $P_v$ follow their imposed directions. \\
    \bottomrule
    \end{tabularx}
\end{table}

\section{Window Level Counterfactual Summary}
\label{app:counterfactual_window_section}

Table~\ref{tab:counterfactual_window} maps each coordinate perturbation to a clinical class and reports the window level pass rate. The pass rate is the fraction of held-out windows for which every output with a prespecified nonzero direction has the expected sign. It is distinct from the patient-level bootstrap interval on the mean response in Table~\ref{tab:sensitivity_full}. Values and 95\% intervals use the same 2000-replicate patient clustered resampling protocol described in Section~\ref{app:evaluation_details}.

\begingroup
\small
\hbadness=10000
\setlength{\tabcolsep}{2.5pt}
\renewcommand{\arraystretch}{1.10}
\setlength{\LTleft}{0pt}
\setlength{\LTright}{0pt}
\setcounter{LTchunksize}{4}
\begin{longtable}{@{}>{\raggedright\arraybackslash}p{1.00in}>{\raggedright\arraybackslash}p{1.30in}>{\raggedright\arraybackslash}p{1.10in}>{\centering\arraybackslash}p{0.68in}>{\centering\arraybackslash}p{0.68in}>{\centering\arraybackslash}p{0.68in}>{\centering\arraybackslash}p{0.68in}@{}}
        \multicolumn{7}{@{}p{6.62in}@{}}{\label{tab:counterfactual_window}\textbf{Table~\thetable: Window level counterfactual summary with clinical class examples.} Pressure changes in mmHg, CO in L/min. Brackets give 95\% patient cluster bootstrap intervals for the mean changes, and pass rate is across windows.}\\
        \toprule
        \makecell[l]{Scenario} & \makecell[l]{Clinical class and examples} & \makecell[l]{Axis perturbation} & \makecell{$\Delta$MAP\\(mmHg)} & \makecell{$\Delta$PP\\(mmHg)} & \makecell{$\Delta$CO\\(L/min)} & \makecell{Window\\pass rate} \\
        \midrule
        \endfirsthead
        \multicolumn{7}{c}{\tablename\ \thetable{} -- continued}\\
        \toprule
        \makecell[l]{Scenario} & \makecell[l]{Clinical class and examples} & \makecell[l]{Axis perturbation} & \makecell{$\Delta$MAP\\(mmHg)} & \makecell{$\Delta$PP\\(mmHg)} & \makecell{$\Delta$CO\\(L/min)} & \makecell{Window\\pass rate} \\
        \midrule
        \endhead
        Vasodilator-like (mild)
          & Vasodilators (nitroglycerin, nitroprusside, nicardipine) \citep{koch1974vasodilator}
          & $\Delta\log\taumodel=-0.25$
          & \makecell{$-15.1$\\{\scriptsize[$-15.4$, $-14.7$]}}
          & \makecell{11.3\\{\scriptsize[9.8, 12.8]}}
          & 0.00 & 99.4\% \\
        Vasodilator-like (strong)
          & Same class as preceding row & $\Delta\log\taumodel=-0.50$
          & \makecell{$-26.3$\\{\scriptsize[$-26.7$, $-25.8$]}}
          & \makecell{18.7\\{\scriptsize[16.3, 21.1]}}
          & 0.00 & 100.0\% \\
        Vasoconstrictor-like (mild)
          & Vasoconstrictive vasopressors (phenylephrine, vasopressin) \citep{holmes2005vasoactive,jentzer2015pharmacotherapy}
          & $\Delta\log\taumodel=+0.25$
          & \makecell{18.4\\{\scriptsize[17.8, 19.0]}}
          & \makecell{$-13.8$\\{\scriptsize[$-16.2$, $-11.3$]}}
          & 0.00 & 91.3\% \\
        Vasoconstrictor-like (strong)
          & Same class as preceding row & $\Delta\log\taumodel=+0.50$
          & \makecell{37.7\\{\scriptsize[36.7, 38.7]}}
          & \makecell{$-9.8$\\{\scriptsize[$-14.2$, $-5.5$]}}
          & 0.00 & 91.3\% \\
        Positive-inotrope-like (mild)
          & Positive inotropy (contractility support, modeled as an increase in flow drive $U$) \citep{overgaard2008inotropes}
          & $\Delta\log U=+0.25$
          & \makecell{14.3\\{\scriptsize[14.1, 14.5]}}
          & \makecell{4.5\\{\scriptsize[3.5, 5.5]}}
          & \makecell{1.41\\{\scriptsize[1.35, 1.46]}} & 70.5\% \\
        Positive-inotrope-like (strong)
          & Same class as preceding row & $\Delta\log U=+0.50$
          & \makecell{32.7\\{\scriptsize[32.3, 33.2]}}
          & \makecell{22.4\\{\scriptsize[21.0, 23.7]}}
          & \makecell{3.22\\{\scriptsize[3.09, 3.34]}} & 95.8\% \\
        Venous pressure decrease
          & Preload reduction (decongestion, venodilation, ultrafiltration) \citep{gelman2008venous,hamzaoui2022central}
          & $\Delta P_v=-3.0$ mmHg
          & \makecell{$-2.7$\\{\scriptsize[$-2.7$, $-2.6$]}}
          & \makecell{$-1.6$\\{\scriptsize[$-1.7$, $-1.4$]}}
          & 0.00 & 100.0\% \\
        Venous pressure increase
          & Same preload class as preceding row & $\Delta P_v=+3.0$ mmHg
          & \makecell{2.7\\{\scriptsize[2.6, 2.7]}}
          & \makecell{1.6\\{\scriptsize[1.4, 1.7]}}
          & 0.00 & 100.0\% \\
        Norepinephrine-like
          & Mixed vasopressor/inotrope (norepinephrine) \citep{holmes2005vasoactive,jentzer2015pharmacotherapy}
          & $\Delta\log\taumodel=+0.25$, $\Delta\log U=+0.10$
          & \makecell{25.5\\{\scriptsize[24.9, 26.1]}}
          & \makecell{$-8.3$\\{\scriptsize[$-11.0$, $-5.4$]}}
          & \makecell{0.52\\{\scriptsize[0.50, 0.54]}} & 100.0\% \\
        Phenylephrine-like
          & Pure $\alpha_1$ agonist (phenylephrine) \citep{thiele2011clinical}
          & $\Delta\log\taumodel=+0.25$, $\Delta\log U=-0.10$
          & \makecell{12.0\\{\scriptsize[11.5, 12.6]}}
          & \makecell{$-16.8$\\{\scriptsize[$-19.0$, $-14.7$]}}
          & \makecell{$-0.47$\\{\scriptsize[$-0.49$, $-0.45$]}} & 90.5\% \\
        Dobutamine-like
          & Inotrope/inodilator (dobutamine) \citep{ruffolo1987pharmacology,overgaard2008inotropes}
          & $\Delta\log\taumodel=-0.10$, $\Delta\log U=+0.25$
          & \makecell{6.1\\{\scriptsize[5.8, 6.5]}}
          & \makecell{4.8\\{\scriptsize[4.0, 5.6]}}
          & \makecell{1.41\\{\scriptsize[1.35, 1.47]}} & 76.2\% \\
        Fluid-bolus-like
          & Preload augmentation (crystalloid, colloid, blood products) \citep{gelman2008venous}
          & $\Delta\log U=+0.25$, $\Delta P_v=+3.0$ mmHg
          & \makecell{17.0\\{\scriptsize[16.8, 17.2]}}
          & \makecell{6.0\\{\scriptsize[4.9, 7.0]}}
          & \makecell{1.41\\{\scriptsize[1.35, 1.47]}} & 100.0\% \\
        Preload-reduction-like
          & Decongestion or venodilation (furosemide, nitroglycerin, ultrafiltration) \citep{gelman2008venous,koch1974vasodilator}
          & $\Delta\log U=-0.25$, $\Delta P_v=-3.0$ mmHg
          & \makecell{$-13.8$\\{\scriptsize[$-14.0$, $-13.6$]}}
          & \makecell{3.1\\{\scriptsize[2.1, 4.1]}}
          & \makecell{$-1.10$\\{\scriptsize[$-1.14$, $-1.05$]}} & 100.0\% \\
        \bottomrule
    \end{longtable}
\endgroup

\end{document}